\documentclass[pra,showpacs,twocolumn,superscriptaddress,aps]{revtex4-1}
\usepackage{amsmath,amssymb,amscd,amsthm,etoolbox}
\usepackage{qcircuit}
\usepackage{float}
\usepackage[caption=false]{subfig}
\usepackage{graphicx}
\graphicspath{{figures/}}
\patchcmd{\maketitle}
  {\finalMaketitle}
  {\finalMaketitle\tableofcontents\vspace{\baselineskip}}
  {}{}
\makeatletter
\def\ps@pprintTitle{%
 \let\@oddhead\@empty
 \let\@evenhead\@empty
 \def\@oddfoot{}%
 \let\@evenfoot\@oddfoot}
\makeatother
\usepackage[utf8]{inputenc}
\usepackage[T1]{fontenc}
\usepackage{lmodern}
\usepackage[inline,shortlabels]{enumitem}
\usepackage{mathtools}
\usepackage{amsfonts}
\usepackage{dsfont}
\usepackage{booktabs}
\usepackage{comment} 
\usepackage[colorlinks,linkcolor=blue,citecolor=blue,urlcolor=blue]{hyperref}
\newcommand{\la}{\lambda}

\newcommand{\si}{\sigma}

\newcommand{\Ga}{\Gamma}

\newcommand{\cC}{{\mathcal C}}

\newcommand{\pd}{\partial}

\newcommand{\ket}[1]{|#1\rangle}
\newcommand{\bra}[1]{\langle#1|} 

\newcommand{\mss}{\kern 1pt}
\renewcommand{\leq}{\leqslant}
\renewcommand{\geq}{\geqslant}
\renewcommand{\le}{\leqslant}

\newcommand{\tends}[1]{\bbuildrel{\hbox to 2em{\rightarrowfill}}_{#1}^{}}

\newcommand{\tr}{\operatorname{tr}}

\newcommand{\qbinom}[3]{{#1\atopwithdelims[]#2}_{\raise 3pt\hbox{$\scriptstyle #3$}}}

\newcommand{\ee}{\end{equation}}
\newcommand{\bea}{\begin{eqnarray}}
\newcommand{\eea}{\end{eqnarray}}

\newcommand{\Int}[1]{\,\mathop{\!#1}\limits^{\lower1ex\hbox{$\scriptstyle\circ$}}{}}
\newtheorem{thm}{Theorem}

\newtheorem{lemma}{Lemma}
\newtheorem{cor}{Corollary}
\theoremstyle{remark}

\newtheorem{fact}{Fact}
\def\clap#1{\hbox to 0pt{\hss#1\hss}}

\begin{document}

\title{Symmetry-resolved entanglement detection using partial transpose moments}
\date\today

\author{Antoine Neven}
\thanks{These authors contributed equally.}
\affiliation{Institute for Theoretical Physics, University of Innsbruck, A–6020 Innsbruck, Austria}

\author{Jose Carrasco}
\thanks{These authors contributed equally.}
\affiliation{Institute for Theoretical Physics, University of Innsbruck, A–6020 Innsbruck, Austria}

\author{Vittorio Vitale}
\affiliation{The Abdus Salam International Center for Theoretical Physics, Strada Costiera 11, 34151 Trieste, Italy}
\affiliation{SISSA, via Bonomea 265, 34136 Trieste, Italy}

\author{Christian Kokail}
\affiliation{Center for Quantum Physics, University of Innsbruck, Innsbruck A-6020, Austria}
\affiliation{Institute for Quantum Optics and Quantum Information of the Austrian Academy of Sciences,  Innsbruck A-6020, Austria}

\author{Andreas Elben}
\affiliation{Center for Quantum Physics, University of Innsbruck, Innsbruck A-6020, Austria}
\affiliation{Institute for Quantum Optics and Quantum Information of the Austrian Academy of Sciences,  Innsbruck A-6020, Austria}

\author{Marcello Dalmonte}
\affiliation{The Abdus Salam International Center for Theoretical Physics, Strada Costiera 11, 34151 Trieste, Italy}
\affiliation{SISSA, via Bonomea 265, 34136 Trieste, Italy}

\author{Pasquale Calabrese}
\affiliation{The Abdus Salam International Center for Theoretical Physics, Strada Costiera 11, 34151 Trieste, Italy}
\affiliation{SISSA, via Bonomea 265, 34136 Trieste, Italy}
\affiliation{INFN, via Bonomea 265, 34136 Trieste, Italy}

\author{Peter Zoller}
\affiliation{Center for Quantum Physics, University of Innsbruck, Innsbruck A-6020, Austria}
\affiliation{Institute for Quantum Optics and Quantum Information of the Austrian Academy of Sciences,  Innsbruck A-6020, Austria}

\author{Beno\^it Vermersch}
\affiliation{Center for Quantum Physics, University of Innsbruck, Innsbruck A-6020, Austria}	
\affiliation{Institute for Quantum Optics and Quantum Information of the Austrian Academy of Sciences,  Innsbruck A-6020, Austria}
\affiliation{Univ.  Grenoble Alpes, CNRS, LPMMC, 38000 Grenoble, France}

\author{Richard Kueng}
\affiliation{Institute for Integrated Circuits, Johannes Kepler University Linz, Altenbergerstrasse 69, 4040 Linz, Austria}

\author{Barbara Kraus}
\affiliation{Institute for Theoretical Physics, University of Innsbruck, A–6020 Innsbruck, Austria}

   \begin{abstract}
   We propose an ordered set of experimentally accessible conditions 
   for detecting entanglement in mixed states.
   The $k$-th condition involves comparing moments of the partially transposed density operator up to order $k$. 
   Remarkably, the union of all moment inequalities reproduces the Peres–Horodecki criterion for detecting entanglement. 
   Our empirical studies highlight that the first four conditions already detect mixed state entanglement reliably in a variety of quantum architectures. 
   Exploiting symmetries can help to further improve their detection capabilities. 
   We also show how to estimate moment inequalities based on local random measurements of single state copies (classical shadows) and derive statistically sound confidence intervals as a function of the number of performed measurements. Our analysis  includes the experimentally relevant situation of drifting sources, i.e.\ non-identical, but independent, state copies.
   
   \end{abstract}

\maketitle

 \section{Introduction}
In the past years, considerable effort led to the building of larger and larger Noisy Intermediate-Scale Quantum (NISQ) devices~\cite{Deutsch2020,Preskill2018,NASEM2020}. For the benchmarking of such devices comes the need for more scalable tools in order to characterize the underlying many-body quantum state (see e.g.~\cite{Eisert} and references therein). For instance, characterizing the entanglement properties of these quantum states is, besides the intrinsic theoretical interest, essential to gauge the performance and verify the proper working of the NISQ devices.

As a first prominent example among the tools to characterize entanglement, there is the concept of entanglement witness~\cite{HHH96}. An entanglement witness is a functional of the quantum density matrix that separates a specific entangled state from the set of all separable states~\footnote{When this functional is linear, it can be identified with an observable whose expectation value can be used to decide whether the target state is entangled or not.}. By contrast, in this work, we shall focus on a superset of the set of separable states: the set of states with positive partial transpose. In other words, we will focus on sufficient conditions for entanglement (equivalently, necessary conditions for separability).

From the numerous theoretical sufficient conditions for entanglement that have been developed in the literature, many cannot be straightforwardly implemented experimentally, mainly because they require the (exponentially expensive) knowledge of the full density matrix~\cite{HHHH09,Guehne_review,AFOV08}. This is for instance the case of the celebrated PPT condition~\cite{Peres96}, which states that a separable state $\rho$ always has a positive semi-definite (PSD) partial transpose (PT) $\rho^\Ga$ for any bipartite splitting of its subsystems. Thus, if $\rho^\Ga$ has (at least) a single negative eigenvalue, then $\rho$ is entangled. The negativity, which resulted from this condition, is a highly used entanglement measure for mixed states~\cite{VW02,Pl05}.

This powerful entanglement condition, which found many applications in theoretical works \cite{cct-2012,cct-2013,castelnovo-2013,Eisler2014B,Wen2016B,rac-2016-2,Blondeau_Fournier_2016,rac-2016}, is difficult to apply in experimental conditions as the PT spectrum is difficult to access. To overcome this challenge, it was shown in Ref.~\cite{EKH20} that valuable information about the PT spectrum can be obtained from a few PT moments ${\rm tr}(\rho^\Ga)^k$ only. Using the first three PT moments, an entanglement condition, called $p_3$-PPT, was proposed and shown to be useful for detecting entangled states in several different contexts. Moments ${\rm tr}(\rho^\Ga)^k$ have the advantage that they can be estimated using shadow tomography~\cite{EKH20} in a more efficient way than if one had to reconstruct $\rho$ via full quantum state tomography. As in other randomized measurements protocols probing entanglement~\cite{VanEnk2012,Elben2018,Elben2018a,Brydges2019,Knips2019,Ketterer2019,huang2020shadow,EKH20,Zhou2020,Ketterer2020,Ketterer2020a}, the classical shadows formalism only requires (randomized) single-qubit measurements in experiments realizing the single-copy state $\rho$. 
In this paper, we follow this idea of using PT moments to build experimentally computable entanglement conditions, and extend the $p_3$-PPT condition in two directions. 

On the one hand, we propose different entanglement detection strategies depending on how many PT moments can be estimated. Starting from the third order moment, we show that the estimation of each higher order moment gives access to an independent entanglement condition. Interestingly, if all the PT moments can be estimated, this set of conditions is then necessary and sufficient for the state to be PPT (i.e. to have a positive semi-definite partial transpose). Of course, the higher the moment, the larger the number of experimental runs needed. In case higher order moments cannot be accessed, we show how to obtain an optimal entanglement condition using PT moments of order up to three.

On the other hand, we investigate the effect of symmetries on this entanglement detection method. As shown in Ref.~\cite{Vi21} for the case of dynamical purification, taking symmetries into account to define symmetry-resolved (SR) versions of the tools usually used to characterize quantum states can enable a finer characterization of some quantum features and even reveal phenomena that cannot be observed without symmetry-resolution. For states preserving an extensive quantity, we define SR versions of the PT-moment inequalities mentioned previously and show that these are indeed better suited to detect the entanglement of such states. Furthermore, we also show that these SR inequalities provide a sufficient entanglement condition for states that do not possess any symmetry.

The conditions derived here are particularly interesting from an experimental (and numerical) point of view, as low moments of (partially transposed) density operators are accessible. We show how source drifts in an experiment can be taken into account and how the quantities which are of interest here can be accurately estimated via local measurements on single copies of the state.

The paper is structured as follows. In Sec.~\ref{sec.summary}, we summarize our results. Our methods to obtain entanglement conditions from PT-moments are presented in  Sec.~\ref{sec.ent_det_PT_moments}. In Sec.~\ref{sec.sr} we study the effect of symmetries and show how to obtain a SR version of these PT-moment inequalities. In Sec.~\ref{sec.Applications}, we apply these inequalities to a variety of physical systems and compare their efficiency for detecting entanglement. Finally, we conclude and give some outlook in Sec.~\ref{sec.conclusion}.

\section{Definitions and summary of results}\label{sec.summary}
In this section, we introduce the basic definitions needed for the entanglement detection criteria below, and summarize in a succinct manner our main results. Given a bipartite state $\rho=\rho_{AB}$, we denote by $\rho^\Ga$ its partial transpose with respect to subsystem $B$. 
We say that $\rho$ is PPT if $\rho^\Ga$ is positive semi-definite, and NPT otherwise. 
All NPT states are entangled, however there are entangled states, known as bound-entangled states, that are not NPT. We focus here on the detection of NPT entangled states. 

We denote the $k$-th order moment of a matrix $M$ by
\begin{equation}\label{eq:pkm}
p_k(M)\equiv\tr M^k.
\end{equation}
We will mostly consider moments of $\rho^\Gamma$, and sometimes use the short-hand notation $p_k\equiv p_k(\rho^\Ga)$. In the presence of symmetries, the partial transpose can be cast in block diagonal form: we denote as $\rho^\Gamma_{(q)}$ the resulting blocks, where $q$ indicates a quantum number, and define the corresponding moments $p_k(\rho^\Gamma_{(q)})$ as from Eq.~\eqref{eq:pkm}.

We start by recalling the $p_3$-PPT condition of Ref.~\cite{EKH20}, i.e. that  any PPT state satisfies
\begin{equation}\label{p3-ppt}
p_3(\rho^\Ga)p_1(\rho^\Ga)\geq(p_2(\rho^\Ga))^2\,.
\end{equation}
Any state violating this condition is NPT and therefore entangled. 
The $p_3$-PPT condition will serve as a reference point below in accessing the predictive power of the new relations. 

In the following, we will establish several sets of necessary (and sometimes also sufficient) PPT conditions, summarized as follows:

{\it i)} the first set of conditions, that we dub $D_n$ conditions, also contains polynomial inequalities in the moments $p_k$ of order up to $k\leq n$. 
The first non-trivial such a condition is $D_3$, and reads: 
\begin{equation}
p_3(\rho^\Gamma)  \geq -\frac{1}{2} (p_1(\rho^\Gamma))^3 + \frac{3}{2} p_1(\rho^\Gamma) p_2(\rho^\Gamma).
\end{equation}
Knowing only the first three moments $p_1(\rho^\Gamma),p_2(\rho^\Gamma)$ and $p_3(\rho^\Gamma)$, this condition is optimal for detecting entanglement if $1/2\leq p_2(\rho^\Gamma)\leq 1$. Knowing moments of order up to the dimension of $\rho^\Gamma$, the set of $D_n$ conditions becomes a necessary and sufficient condition for NPT entanglement;

{\it ii)} the second set of conditions, dubbed $\textrm{Stieltjes}_n$, involves inequalities among the moments $p_k$ of order up to $n$. 
The condition $\textrm{Stieltjes}_3$ is equivalent to $p_3$-PPT, while $\textrm{Stieltjes}_5$ reads:
\begin{equation}
\det
\begin{pmatrix}
p_1 & p_2 & p_3 \\
p_2 & p_3 & p_4 \\
p_3 & p_4 & p_5 \end{pmatrix} \geq 0
\end{equation}
and similar conditions are obtained including higher moments; 

{\it iii)} in case high-order moments are difficult or too expensive to access, we also show how to obtain an optimized, necessary condition for PPT using only PT moments of order up to three. We call this condition $D_3^\textrm{opt}$;

{\it iv)} all of the above conditions can be cast in terms of $\rho^\Gamma_{(q)}$, in which case we add the prefix SR (for symmetry-resolved). For instance, the SR-$p_3$-PPT condition for sector $q$ reads
\begin{equation}
p_3(\rho^\Ga_{(q)})p_1(\rho^\Ga_{(q)})\geq(p_2(\rho^\Ga_{(q)}))^2.
\end{equation} 
Since these conditions are sensitive to the presence of negative eigenvalues in a specific symmetry sector, they are typically much more sensitive than their non-SR counterparts, as illustrated in Fig.~\ref{fig.summary}. 

In the SR case, it is worth mentioning that also the SR-$D_2$ condition,
\begin{equation}
p_2(\rho^\Ga_{(q)}) \leq (p_1(\rho^\Ga_{(q)}))^2\,,
\end{equation}
is non-trivial;

\begin{figure*}[t]
  \begin{minipage}{\linewidth}
    \centering
    \includegraphics[width=\linewidth]{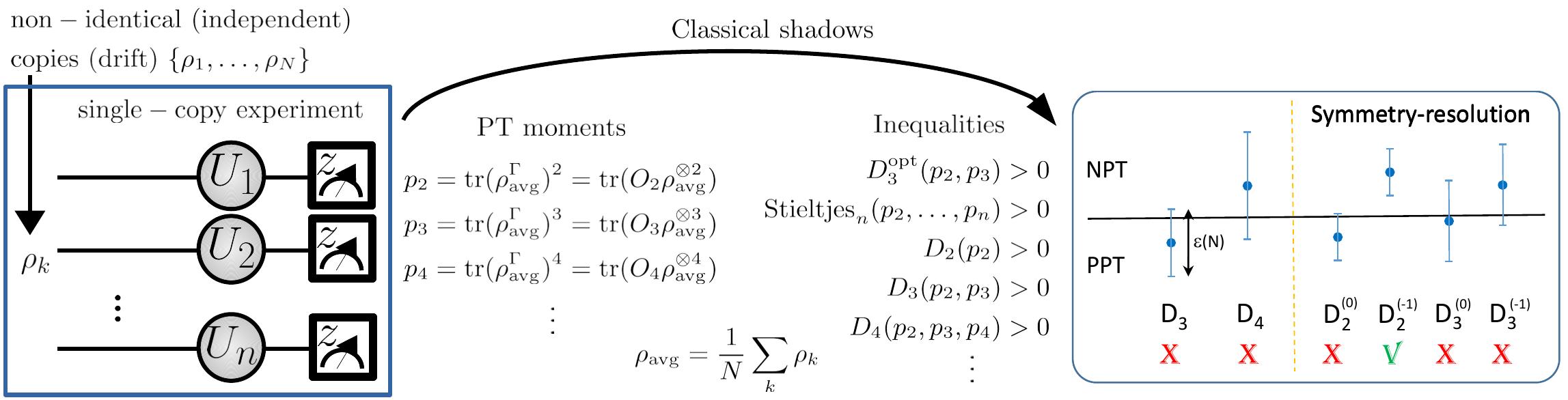}
    \caption{{\em An illustration of the proposed method for entanglement detection}. We assume the experimentally relevant situation of a source producing non-identical but independent copies $\{\rho_1,\ldots,\rho_N\}$ ("drift"). Randomly-chosen unitaries $U_i$ are applied to the qubits of each copy and then measured in the standard basis. Using classical shadows~\cite{huang2020shadow}, these measurement outcomes are post-processed to obtain the moments $p_j={\rm tr}(\rho_{\rm avg}^\Gamma)^j$. As explained in the main text, we combine those moments to derive inequalities whose violation implies that the state $\rho_{\rm avg}$ is NPT, thus showing that at least one of the states $\rho_k$ produced by the source is entangled. We also show how symmetry-resolution techniques can be used to enhance the entanglement detection capabilities.}
    \label{fig.summary}
  \end{minipage}
\end{figure*}

{\it v)} we show how SR conditions can, in fact, be applied to arbitrary states, via application of a proper transformation on the density matrix of interest. In practice, this transformation is effectively carried out in the post-processing step of the classical shadows;

{\it vi)} as illustrated in Fig.~\ref{fig.summary}, the uncertainty in estimating the moments can be bounded, in principle, using the classical shadows formalism. Here, we show how to combine those bounds to provide rigorous confidence intervals for SR-$D_2$, which considerably strengthen the impact of our results in real experiments.

\section{Entanglement detection from partial transpose moments}\label{sec.ent_det_PT_moments}

In this section, we present entanglement conditions based on PT moments. We extend the idea behind the $p_3$-PPT condition (c.f.~Eq.~\eqref{p3-ppt}) of Ref.~\cite{EKH20} in two directions. On the one hand, we present a set of inequalities involving all the PT moments which provides a necessary and sufficient condition for the underlying state to be PPT. In addition, each condition of this set is itself a necessary PPT condition. On the other hand, we show how to optimize such entanglement conditions when only few low-order PT moments are accessible.

The idea behind this set of conditions is to use Descartes' rule of signs on the characteristic polynomial of a Hermitian matrix to obtain a set of moment inequalities that has to be satisfied by any PSD matrix. Applied to the partially transposed matrix $\rho^\Ga$, such conditions can then be used to detect the entanglement of NPT quantum states. More precisely, using the definition of the elementary symmetric polynomials on $d$ variables,
\begin{equation}\label{eq.ElPoly}
    e_i(x_1,\dots,x_d) = \sum_{1\leq j_1 < \cdots < j_i \leq d} x_{j_1} \cdots x_{j_i},
\end{equation}
for $i=1,\dots,d$, and $e_0(x_1,\dots,x_d)=1$, we derive in Appendix~\ref{App.Des} the following lemma.

\begin{lemma}\label{lem.con}
   A Hermitian matrix $A$ of dimension $d$ is PSD if and only if  $e_i(\la_1,\ldots,\la_d) \geq 0$ for all $i=1,\ldots,d$, where $\la_1,\ldots,\la_d$ are the eigenvalues of $A$, and $e_i$ denote the elementary symmetric polynomials (Eq.~\eqref{eq.ElPoly}).
\end{lemma}

Using Newton's identities, which relate the elementary symmetric polynomials, $e_k$, in the eigenvalues of $A$ to the moments of $A$ through the recursive formula
\begin{equation}\label{sym}
  k\mss e_k=\sum_{i=1}^k(-1)^{i-1}e_{k-i} \; p_i(A),
\end{equation}
each inequality $e_i \geq 0$ can be transformed into an inequality involving moments of $A$ of order up to $i$. We denote by $D_i$ these moments inequalities. As an illustration, the conditions $D_1$ to $D_4$ read 
\begin{align}
    & p_1(A) \geq 0, \label{eq:D1} \\
    & p_2(A) \leq (p_1(A))^2, \label{eq:D2} \\
    & p_3(A)  \geq -\frac{1}{2} (p_1(A))^3 + \frac{3}{2} p_1(A) p_2(A), \label{eq:D3} \\
    & p_4(A) \leq \frac{1}{2} \left( (p_1(A))^2 - p_2(A) \right)^2 - \frac{1}{3} (p_1(A))^4\notag\\
    &\hspace{45mm}+ \frac{4}{3} p_1(A) p_3(A), \label{eq:D4}
\end{align}
respectively. One has $p_1(\rho^\Ga)=1$ for any quantum state $\rho$, implying that $D_1$ is trivially satisfied. Similarly, since $p_2(\rho^\Ga)$ is equal to $p_2(\rho)$ (i.e., to the purity of $\rho$) for any quantum state $\rho$, the inequality $D_2$ is also trivially satisfied. Therefore, when $\rho$ is a quantum state, the first non-trivial inequality for $\rho^\Ga$ is $D_3$. As will be shown in the next section, it is sometimes more efficient (in order to detect entanglement) to apply these inequalities to projections of $\rho^\Ga$ onto specific subspaces, rather than to $\rho^\Ga$ itself. We would like to stress here that, in that case, the argument above does not hold, so that the inequality $D_2$ is not trivially satisfied and can already reveal the presence of entanglement (see Sec.~\ref{sec.sr}). 

When applied to $\rho^\Ga$, Lemma~\ref{lem.con} and Newton's identities~\eqref{sym} can thus be used to detect NPT entangled states from PT moments only. From an experimental point of view, this is an important aspect of this entanglement detection scheme, as PT moments are experimentally more affordable to estimate than, for instance, the whole  spectrum of $\rho^\Ga$. As PT moments are more expensive to be estimated the higher the order, these inequalities should be considered starting from those involving the lowest moment orders. Even though showing that a state is NPT with this method can in principle require the knowledge of all the PT moments, we will provide many experimentally relevant instances where entanglement can be effectively detected from low-order moments even in the presence of errors. To this end, we provide confidence intervals for the quantities of interest granting that a certain inequality is violated with high probability (see Theorem~\ref{thm:r1} and, e.g, Appendix~\ref{App.shadows}).

Similarly, let us mention here that necessary and sufficient conditions for a matrix to be PSD can be expressed as different sets of polynomial inequalities in its moments. One of such sets can be deduced from the well-known (truncated) {\em Stieltjes moment problem} (see Appendix~\ref{App.Sti}). In Sec.~\ref{sec.Applications}, we illustrate the usefulness of these inequalities by applying them to the entanglement detection of the ground state of the XXZ model (c.f. Fig~\ref{fig:XXZ_sim}). Let us finally also mention that, from a few moments of a Hermitian matrix, one can also obtain bounds on the distance between this matrix and the PSD cone \cite{Gemma20}. 

\subsection{Optimized condition for low-order moments}
\label{sec.optimization}
Due to (experimental) constraints, it might  not be possible to determine all, but only a few, PT moments. This is why, we show here how to optimize necessary PPT conditions using only PT moments of order up to three. From the previous sections, we already have two examples of such conditions, namely the $p_3$-PPT and $D_3$ conditions. As illustrated in Fig.~\ref{fig.comparing}, the $p_3$-PPT ($D_3$) condition is tighter than $D_3$ ($p_3$-PPT) for states with purity larger (smaller) than $1/2$, respectively. As the low-order moments are easier to access experimentally, we now address the question about the optimal inequality involving PT moments of order up to three. 

\begin{figure}
  \centering
  \hspace{-.2mm}\includegraphics[width=\columnwidth]{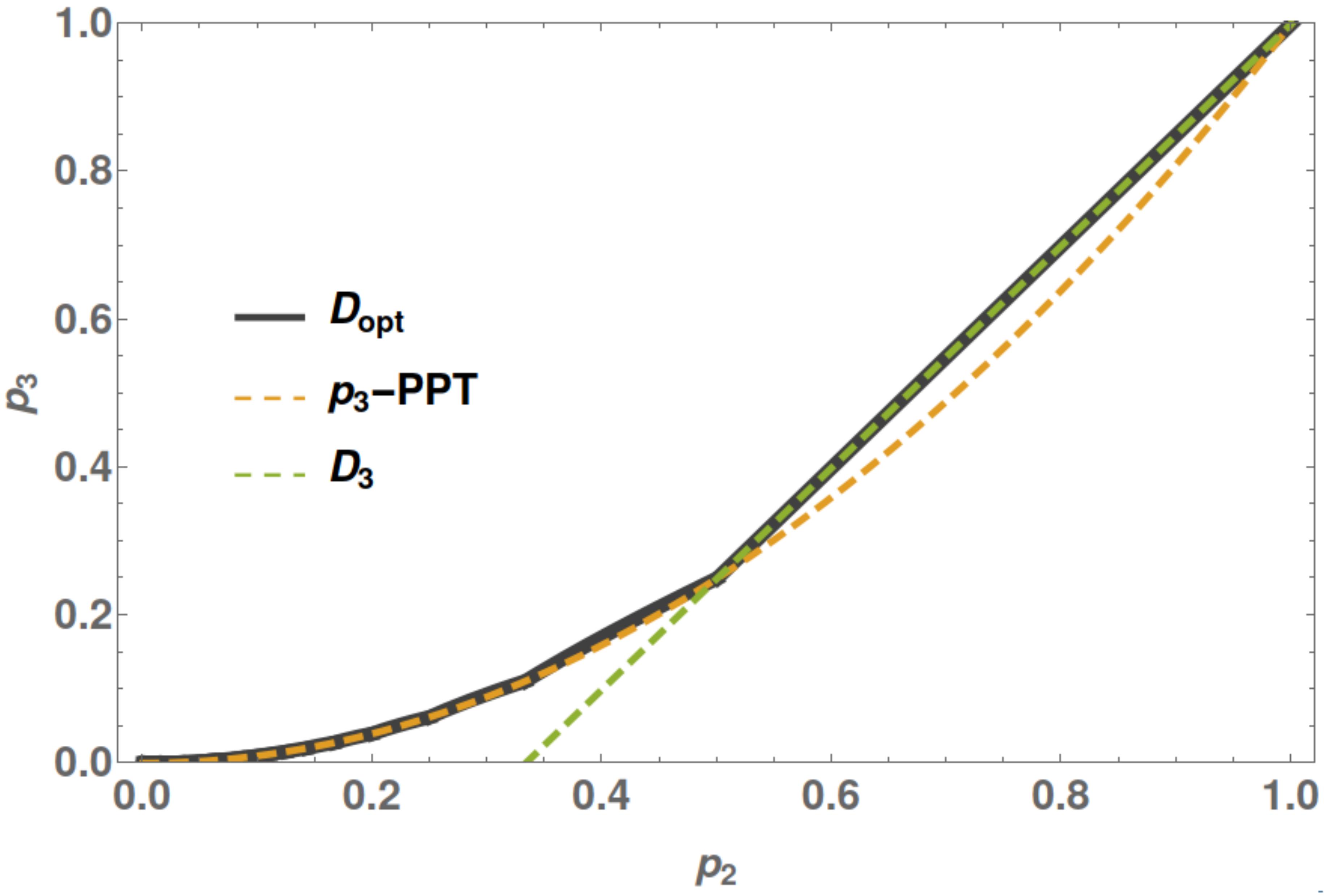}
  \caption{Plot of the value of the third moment $p_3$ saturating the $p_3$-PPT (dashed orange curve), the $D_3$ (dashed green line) and the optimal $D_3^{\rm opt}$ (thick black curve) conditions as a function of the second moment $p_2$ for a normalized Hermitian matrix. According to the $p_3$-PPT condition, any state $\rho$ with a value of $p_3(\rho^\Ga)$ below the dashed orange curve is entangled. Similarly, the condition $D_3$ shows that any state $\rho$ with a value of $p_3(\rho^\Ga)$ below the dashed green line is entangled. From this plot, it is clear that, for $p_2(\rho^\Ga)>1/2$, all entangled
    states detected by the $p_3$-PPT condition are also detected by $D_3$, which coincides with $D_3^{\rm opt}$ in this case. When $p_2(\rho^\Ga)<1/2$, the $p_3$-PPT condition is then {\em stronger} than $D_3$, and $D_3^{\rm opt}$ represents a slight improvement over the $p_3$-PPT condition. As illustrated in Sec.~\ref{sec.Applications}, this slight improvement can nevertheless be important for the detection of physically relevant states.}
  \label{fig.comparing}
\end{figure}

To answer this question, we use the following approach. For fixed values $p_1$ and $p_2$ of the first two moments, we determine the minimal value $p_3^{\rm min}$ that the third moment can reach for any PSD matrix~\footnote{Note that we want here to minimize $p_3$ because it is an odd moment (for which negative eigenvalues would have the tendency to decrease the value of the moment). For an even moment, we would instead maximize the value of this moment over PSD matrices. This is also reflected in the $D_n$ conditions~\eqref{eq:D2}--\eqref{eq:D4}, where the inequality sign alternates between even and odd values of $n$.}. From this bound, we know that any Hermitian matrix with a smaller third moment is necessarily not PSD. Naturally, we restirc ourselves to values of $p_1$ and $p_2$ which are compatible with a PSD matrix, and therefore satisfy Eqs.~\eqref{eq:D1} and~\eqref{eq:D2} \footnote{Recall that for the partial transpose of a density operator this is always fulfilled. }. 

Given a $d\times d$ PSD matrix $A$, with non-zero eigenvalues $\lambda_1,\dots,\lambda_r$, for some $r\in[1,d]$, this optimization can be performed with the help of Lagrange multipliers. As shown in Appendix~\ref{App.Lagrange}, this leads to solutions with only two distinct eigenvalues $\lambda_a,\lambda_b$ with multiplicity $r_a,r-r_a$, respectively, for $r_a\in[1,r]$. Assuming, without loss of generality, that $\lambda_a \geq \lambda_b$, the optimization of $p_3$ leads then to $r_a=r-1$. For each value of $r$, the optimal value of $p_3$ can be easily determined in the interval $[1/r,1/(r-1)]$. For $r=2$ this leads to $D_3$ whereas for $r>2$ one obtains an optimal value of $p_3$ which is slightly better than $p_3$-PPT. Observe that $p_3^{\rm min}(p_2)$ is a piece-wise function and the derivative $\pd p_3^{\rm min}/\pd p_2$ is discontinuous at points $p_2=1/r$ (see Fig.~\ref{fig.comparing}).

\section{Symmetry-Resolved entanglement detection}\label{sec.sr}
Symmetries, as they often occur in physical situations, can be exploited to observe relevant phenomena (see e.g. Refs.~\cite{GoldsteinSela2018,xas-18,Feldman2019,BRC19,Tan2020,fraenkel2020symmetry,Murciano2020,Azses2020,Turkeshi2020,murciano2021symmetry,pbc-2021,Vi21}). Here, we use symmetries to ease the detection of entanglement. More precisely, we apply the previously developed tools to symmetric states, which will lead to conditions of entanglement involving much lower moments of the partial transpose projected onto certain subspaces. Despite the fact that these quantities differ significantly from the moments of $\rho^\Ga$, we will show later on that they can nevertheless be estimated using the framework of classical shadows. 

We consider a bipartite state $\rho=\rho_{AB}$, with subsystems $A$ and $B$ containing $n$ and $m$ qubits, respectively. We assume that this state commutes with $\sum_{i=1}^{n+m}Z_i$, or similarly with the total number operator $\mathcal{N}=\mathcal{N}_A+\mathcal{N}_B$. Here and in the following, we denote by $X,Y,Z$ the Pauli operators. Obviously, such a state has a block diagonal form, i.e.,
\begin{equation}\label{eq:decom}
\rho=\bigoplus_{q=0}^{n+m}{\rho_{(q)}}=\sum_q Q_q\rho Q_q,
\end{equation}
where each block (or sector) is labeled by an eigenvalue $q\in\{0,1,\ldots,n+m\}$ of the operator $\mathcal{N}$ and has support in the corresponding eigenspace. Here, 
\bea
    Q_{q}=\sum_{a+b=q}\Pi_a(A)\otimes\Pi_b(B) \label{def:Q},
\eea     
with \[
\begin{aligned}
    \Pi_k(A)&=\sum_{i_1+\cdots+i_n=k}\ket{i_1\cdots i_n}\bra{i_1\cdots i_n}\,,
 \end{aligned}
 \]    
and similarly for $B$. 
It has been shown~\cite{CGS18} that, for this type of symmetry, the partial transpose $\rho^\Ga$ is also block diagonal, but in a different basis. In fact, 
$\rho^\Ga=\oplus_{q=-m}^{n} \rho^\Ga_{(q)}= \sum_q P_q\rho^\Ga P_q$, where $P_q$ is the projector onto the eigenspace of
$\mathcal{N}_A-\mathcal{N}_B$ with eigenvalue $q\in\{-m,-m+1,\dots,n\}$~\footnote{This can be easily seen as follows. Consider a matrix element $\rho_{ab,a'b'}\ket{ab}\bra{a'b'}$ of $\rho$ with eigenvalue $i$ of $\mathcal{N}_A+\mathcal{N}_B$. Precisely, let us write $\mathcal{N}_A\ket a=n_a\ket a$, $\mathcal{N}_A\ket{a'}=n_{a'}\ket{a'}$, $\mathcal{N}_B\ket b=n_b\ket b$, and $\mathcal{N}_B\ket{b'}=n_{b'}\ket{b'}$ with $n_a+n_b=n_{a'}+n_{b'}=i$. After partial transposition, $\rho_{ab,a'b'}\ket{ab}\bra{a'b'}\mapsto\rho_{ab,a'b'}\ket{ab'}\bra{a'b}$. For our particular case, $\mathcal{N}_B=\mathcal{N}_B^\Ga$ and one can see that $(\mathcal{N}_A-\mathcal{N}_B)\ket{ab'}=(n_a-n_{b'})\ket{ab'}$ and $(\mathcal{N}_A-\mathcal{N}_B)\ket{a'b}=(n_{a'}-n_b)\ket{a'b}$ with $n_a-n_{b'}=n_{a'}-n_b$. This shows that matrix elements within a block of $\rho$ are mapped, via partial transposition, to matrix elements within a block of $\rho^\Ga$.}, i.e. 
\bea
    P_{q}=\sum_{a-b=q}\Pi_a(A)\otimes\Pi_b(B) \label{def:P}.
\eea
The size of the sector corresponding to the eigenvalue $q$ in the block-decomposition of $\rho^\Ga$ is given by
\[
\tr P_q=\sum_{a-b=q}\binom na\binom mb=\binom{n+m}{q+m}\,.
\]

When the partial transpose of a density matrix has a block structure, it is naturally PSD iff each block is itself a PSD matrix. Therefore, one can apply the conditions of the previous section directly to the blocks $\rho^\Ga_{(q)}$ of the partial transpose. For the $p_3$-PPT condition, the corresponding symmetry-resolved (SR) inequalities are simply 
\[
p_3(\rho^\Ga_{(q)})p_1(\rho^\Ga_{(q)})\geq(p_2(\rho^\Ga_{(q)}))^2
\]
for all $q=-m,-m+1,\ldots,n$. Any violation of a PSD condition in one of the blocks is then sufficient to show that $\rho^\Ga$ has at least one negative eigenvalue and that $\rho$ is therefore entangled.

When using the $D_i$ conditions, symmetry-resolution can be a significant advantage (see e.g. Sec.~\ref{sec.Applications}). First, the necessary and sufficient PSD conditions involve moments of order at most equal to the dimension of the largest block, that is $\binom {n+m}{\lfloor (n+m)/2 \rfloor}$, which is necessarily smaller than the dimension of the density matrix itself. Second, since a block $\rho_{(q)}^\Ga$ of $\rho^\Ga$ is (in general) not the partial transpose of any positive matrix~\footnote{This is, there could be no $\si>0$ such that $\rho_{(q)}^\Ga=\si^\Ga$.}, the inequality: 
\begin{equation}\label{eq:SR_d2}
p_2(\rho_{(q)}^\Ga)\leq(p_1(\rho_{(q)}^\Ga))^2
\end{equation}
is not necessarily satisfied. This implies that moments of order two can already be sufficient to detect entanglement. 

As stressed in the introduction, using PT-moment inequalities to detect entanglement is particularly interesting from an experimental point of view, because such PT moments can be estimated, for instance using shadow tomography~\cite{EKH20}. As we show in the following lemma, the shadow tomography protocol used in Ref.~\cite{EKH20} can also be used to estimate moments of blocks of the partial transpose (which differ significantly from the PT moments) by slightly modifying the non-linear observable that has to be measured.

\begin{lemma}\label{lem.Li}
  Given a symmetric state $\rho=\sum_iQ_i\rho Q_i$, for each eigenvalue $i$ of $\mathcal{N}_A-\mathcal{N}_B$, it holds that
  \[
    \tr(P_i\rho^\Ga P_i)^k=\tr(L_i^{(k)}\rho^{\otimes k})
  \]
  where the operators $L_i^{(k)}$ are given by
  \begin{multline*}
    L_i^{(k)}=\left(\sum_{a-b=i}\Pi_a(A_1)\otimes1\otimes\cdots\otimes1\otimes\Pi_b(B_k)\right)\cdot \\ \tilde
    S(A_1,\ldots,A_k)\otimes S(B_1,\ldots,B_k)
  \end{multline*}
  with 
  \[
    \begin{aligned}
      \tilde
      S(A_1,\ldots,A_k)&=\sum_{a_1}\cdots\sum_{a_k}\ket{a_ka_1\cdots
        a_{k-1}}\bra{a_1\cdots a_k}\,,\\
      S(B_1,\ldots,B_k)&=\sum_{b_1}\cdots\sum_{b_k}\ket{b_2\cdots
        b_kb_1}\bra{b_1\cdots b_k}.
    \end{aligned}
  \]
  Here, the sum over each $a_i$ ($b_i$) runs from $1$ to $2^n$ ($2^m$) respectively. 
\end{lemma}
\begin{proof}
  Denoting by $\rho_{ab,a'b'}$ the entries of the operator and using that  $\rho_{ab,a'b'}=\rho^\Ga_{ab',a'b}$, it is straightforward to see that
  \begin{multline}\label{toshow}
    \tr_{\setminus A_1B_k}\left(\tilde
      S(A_1,\ldots,A_k)\otimes S(B_1,\ldots,B_k)\rho^{\otimes k}\right)\\
      =((\rho^\Ga)^k)^\Ga\,.
  \end{multline}
  Then, we have that
  \[
  \begin{aligned}
    \tr(L_i^{(k)}\rho^{\otimes
      k})=\tr\left(P_i\,((\rho^\Ga)^k)^\Ga\right)&=\tr\left(P_i^\Ga\,(\rho^\Ga)^k\right)\\
      &=\tr\left(P_i\,(\rho^\Ga)^k\right).
  \end{aligned}
  \]
  Here, the first equality follows from the definition of $L_i^{(k)}$ and Eq.~\eqref{toshow}; the second, from $\tr(RS^\Ga)=\tr(R^\Ga S)$ for any two matrices $R,S$; and the
  third, from $P_i^\Ga=P_i$. Finally, using that $P_i=P_i^2$ are
  orthogonal projectors, the cyclic property of the trace, and the
  block structure of $\rho^\Ga$, we have
  \[
    \tr\left(P_i\,(\rho^\Ga)^k\right)=\tr\left(P_i\,(\rho^\Ga)^kP_i\right)=\tr\left(P_i\,\rho^\Ga
      P_i\right)^k,
  \]
  which completes the proof. 
  
\end{proof}

\subsection{Classical shadows}
\label{sec:class_shadows}
Classical shadows are a convenient formalism to reason about predicting properties of a quantum system based on randomized single-qubit measurements performed in single-copy experiments~\cite{huang2020shadow}. We refer the reader to~\cite{huang2020shadow} and Appendix~\ref{App.shadows} for an introduction to classical shadows. The original classical shadow formalism is contingent on noiseless measurements and sources that produce {\em iid} states. Subsequently, it was shown that classical shadows can also handle noisy measurements~\cite{flammia2020shadow,koh2020shadow}. As we show in detail in Appendix~\ref{App.shadows}, the formalism can also be extended to take non-identical, but independent, state preparations ("drifts") into account, i.e. the source produces the states $\{\rho_1, \rho_2,\ldots, \rho_N\}$. In this case, each snapshot $\hat\rho_i$ will have a different expectation value and
\[
\frac1N\sum_{i=1}^N\hat\rho_i\to\frac1N\sum_{i=1}^N{\mathbb E}\hat\rho_i=\frac1N\sum_{i=1}^N\rho_i=:\rho_{\rm avg}\,,
\]
i.e, the average of the snapshots converges to the average state. Since different snapshots are statistically independent, it turns out that one can estimate linear functions, say $\mathrm{tr}(O \rho_{\rm avg})$.

These ideas regarding the prediction of linear observables do extend to higher order polynomials. Here, we restrict ourselves to the quadratic case~\cite{huang2020shadow} involving up to second order moments of a block of the partial transpose. An extension to higher order polynomials is conceptually straightforward, but can become somewhat tedious to analyze~\cite{EKH20}. Let us fix a block label $i$ and consider the second order moment inequality restricted to this block:
\[
\begin{aligned}
D_2^{(i)}(\rho) &= \left( \mathrm{tr} \left( P_i \rho^\Gamma P_i \right) \right)^2 - \mathrm{tr} \left( P_i \rho^\Gamma P_i \right)^2\\
&=\left(p_1( P_i\rho^\Gamma P_i)\right)^2-p_2(P_i\rho^\Gamma P_i)={\rm tr}(Q\rho\otimes\rho)\,,
\end{aligned}
\]
where, using Lemma~\ref{lem.Li},
\[
Q = L_i^{(1)}\otimes L_i^{(1)}-L_i^{(2)}\,.
\]
Recall that $D_2^{(i)}(\rho)<0$ implies that $\rho$ is entangled. Here, we will provide confidence intervals in the estimation of $D_2^{(i)}(\rho_{\rm avg})$ given a fixed number of measurements. Observe that $D_2^{(i)}(\rho_{\rm avg})<0$ implies that there is at least a value $k\in\{1,2,\ldots,N\}$ such that $D_2^{(i)}(\rho_k)<0$. Thus, $D_2^{(i)}(\rho_{\rm avg})<0$ implies that the source is able to produce entangled states.

Let us finally introduce (see Appendix~\ref{App.shadows} for details) the empirical average, over $N(N-1)$ pairs of independent snapshots, 
\[
\widehat{D}_2^{(i)}=\frac{1}{N(N-1)} \sum_{i \neq j} \mathrm{tr} \left( Q \hat{\rho}_i \otimes \hat{\rho}_j \right)\,.
\]
We can fix the desired approximation accuracy $\epsilon$ and a probability-of-error threshold $\delta$ to obtain a lower bound on the measurement budget $N$. For simplicity, let us consider a non-trivial sector (i.e. $q\neq -m,n$) and assume $n+m\geq 4$. Then one has the following theorem (see Appendix~\ref{App.shadows} for a slightly better bound).

\begin{thm}[Error bound for $D_2^{(i)}$]\label{thm:r1}
Fix $\epsilon\in(0,1)$ (accuracy of the approximation), $\delta \in (0,1)$ (probability-of-error threshold), a bipartition $AB$, as well as a symmetry sector $i$. Suppose that we perform
\[
N \geq \frac{2^{n+m}\mathrm{tr}(P_i)}{\epsilon^2 \delta} \frac{1}{2}
\left(4 + \sqrt{4^2 +2 \frac{\epsilon^2 \delta 2^{n+m}}{\mathrm{tr}(P_i)}}
\right)+1
\]
randomized, single-qubit measurements on independent states.
Then,
\begin{equation*}
\left| \widehat{D}_2^{(i)}-D_2^{(i)}(\rho_{\mathrm{avg}}) \right| \leq \epsilon \quad \text{with prob.\ (at least) $1-\delta$.}
\end{equation*}
\end{thm}

Finally, let us stress that this error bound addresses the estimation of $D_2^{(i)}(\rho_{\rm avg})$ in terms of a single U-statistics estimator. The poor scaling in $1/\delta$ can be exponentially improved by dividing the classical shadow into equally-sized batches and performing a median-of-U-statistics estimation instead~\cite{huang2020shadow}: $1/\delta \to \text{const} \times \log (1/\delta)$. However, numerical experiments conducted in Ref.~\cite{EKH20} suggest that this trade-off is only worthwhile if one attempts to predict many properties with the same data set.

\subsection{SR inequalities applied to states without symmetries}
In the last part of this section, we show that the SR inequalities can also be used to detect the entanglement of arbitrary states, including those that do not have any symmetry. 

The reason for that is that there 
exists a local channel $\mathcal{C}$ that transforms any state $\rho$ into a state $\sigma \equiv \mathcal{C}(\rho)$ that has the desired block structure. The channel can be realized with local operations assisted by classical communication, and can thus not generate entanglement. Therefore, the initial state $\rho$ must be at least as entangled as the final block diagonal state $\sigma$. This statement holds for any entanglement measure. As a consequence, if entanglement is detected in $\sigma$ (which can be investigated using the symmetry-resolved tools), then $\rho$ is necessarily also entangled. In other words, looking at the entanglement of $\sigma$, the "block-diagonalized" version of $\rho$, provides a sufficient condition of entanglement for $\rho$. This condition is not necessary as it could be that the channel $\mathcal{C}$ destroys all the entanglement of $\rho$.

The local channel that can be used for this approach is the following:
\begin{equation}
\cC:\rho\to\cC(\rho)=\frac1{2^k}\sum_{i=0}^{2^k-1}U_i^{\otimes (n+m)}\rho\,(U_i^\dagger)^{\otimes(n+m)} \label{eq:channel}
\end{equation}
where $k=\lfloor\log(n+m)\rfloor+1$ and $U_i=Z^{i/2^k}$. The fact that this channel maps $\rho$ to a state $\sigma$ that is block-diagonal in the number-of-excitations basis can easily be seen as follows. First, observe that for any $j\in\{0,\dots, 2^{(n+m)}-1\}$, the computational basis state $|j\rangle$  is an eigenvector of $U_i$, associated to an eigenvalue, $(-1)^{|j| i/{2^k}}$, that essentially depends on $|j|$, the number of excitations of $|j\rangle$. Therefore, we have
\begin{equation}
    \sigma = \frac{1}{2^k}\sum_{j,j^\prime=0}^{2^{(m+n)}-1}\,\,{ \sum_{i=0}^{2^k-1} (-1)^{\frac{i(|j|-|j^\prime|)}{2^k}} \rho_{j,j^\prime} |j\rangle \langle j^\prime| }.
    \label{eq:channelC}
\end{equation}
For any $j$ and $j^\prime$ having different number of excitations, i.e. such that $|j| \neq |j^\prime|$, the sum over $i$ in Eq.~\eqref{eq:channelC} vanishes, explaining why $\sigma$ is diagonal in the number-of-excitations basis.

As can be seen from the argument above, the non-zero elements of $\sigma$ are all equal to the corresponding elements of $\rho$. This implies that the channel can effectively be replaced by a sum of projectors onto all the number-of-excitations sectors. From an experimental point of view, the practical implementation of this channel can thus be circumvented by using the observables of Lemma~\ref{lem.Li} in the post-processing of the classical shadows. 

\section{Applications}
\label{sec.Applications}
In this section, we apply and compare the entanglement conditions presented in the previous sections on various physical systems. For the systems possessing a symmetry as discussed in Sec.~\ref{sec.sr}, we highlight some of the advantages that can result from considering symmetry-resolved entanglement detection tools.

\subsection{Entanglement detection in quench dynamics}\label{sec.sim}

We begin by considering the situation of quench dynamics, where entanglement emerges from the dynamics of a many-body Hamiltonian. We consider the model presented in Ref.~\cite{Vi21}, where the interplay between coherent dynamics with $U(1)$ symmetry and dissipation leads to a dynamics of `purification'. Here, we will use the same formalism to show how entanglement is generated at short times, and can be detected via the symmetry-resolved versions of the  $D_2$ and $p_3$-PPT conditions. In the next subsection, we will consider an analogous experimental situation obtained with trapped ions~\cite{Brydges2019}.

Our model is described by a master equation
\begin{equation}\label{eq:me}
    \partial_t \rho = -\frac{i}{\hbar} [H_{XY}, \rho] + \sum_j \gamma \left[\sigma^-_j\rho \sigma^+_j -\frac{1}{2}\{\sigma_j^+\sigma_j^-, \rho\}\right ], 
\end{equation}
with the lowering and raising operators \mbox{$\sigma_j^- = (X_j-i Y_j)/2$}, $\sigma_j^+ = (X_j+i Y_j)/2$, and the Hamiltonian 
\begin{equation}
  H_{XY} = \frac{\hbar}{2}\sum_{i<j}J_{ij} (X_iX_j+Y_iY_j )
  \label{eq:Hxy}
\end{equation}
and where $\gamma$ is the spontaneous emission rate. Here, we consider open boundary conditions. The hopping  between spins $i$ and $j$ is described by the coefficient $J_{ij}$ and
we consider this subsection  nearest-neighbor hopping $J_{ij}=J\delta_{i+1,j}$ with  strength $J$. The initial state is the N\'eel state $\rho(0)=\ket{\psi(0)}\bra{\psi(0)}$, with $\ket{\psi(0)}=\ket{\!\downarrow \uparrow}^{\otimes N/2}$.
As shown in Ref.~\cite{Vi21}, the time evolved state $\rho(t)$ of the $N$ spin system has the block diagonal form of Eq.~\eqref{eq:decom}. Moreover, the partially transposed matrix w.r.t a partition $A$, $\rho^\Gamma$ is also block diagonal with blocks $\rho^\Gamma_{(q)}$. Here, the index $q$ represents the difference between the number of spin excitations in $A$ and the one in the complement partition $B$ (see also Sec.~\ref{sec.sr}).
 
 As we are interested in short time dynamics, we can solve Eq.~\eqref{eq:me} in first order in perturbation theory, which is valid for $t\ll 1/J,1/\gamma$. Considering for concreteness a half-partition $A$, made of the first $N_A$ sites, we obtain a block with a negative eigenvalue~\cite{Vi21}. Assuming for simplicity $N_A=N/2$, $N_A$ even, we obtain
\begin{align}
\rho^\Gamma_{(-1)}(t)  = & \gamma t \sum_{m=1}^{N_A/2} \sigma^-_{2m} \rho(0) \sigma_{2m}^+\\&
+Jt \left(-i   \sigma^+_{N_A+1} \rho(0) \sigma_{N_A}^+ +\mathrm{h.c}\right) .
\end{align}
The presence of a negative eigenvalue in this sector can be detected from the value of the moments
\begin{eqnarray}
p_1(\rho^\Gamma_{(-1)}(t)) &=& \frac{\gamma N_A t}{2}, \\ 
p_2(\rho^\Gamma_{(-1)}(t)) &=&  2J^2t^2, \\
p_3(\rho^\Gamma_{(-1)}(t)) &=& 3\gamma J^2t^3,
\end{eqnarray}
in leading order in $J\gg \gamma N_A$. In particular, the $p_3$-PPT ratio
\begin{eqnarray}
\frac{p_3(\rho^\Gamma_{(-1)}(t))p_1(\rho^\Gamma_{(-1)}(t))}{p_2(\rho^\Gamma_{(-1)}(t))^2} = \frac{3 \gamma^2 N_A }{8J^2} \ll 1,
\label{eq:p3PPTperturb}
\end{eqnarray}
and the $D_2$ condition
\begin{eqnarray}
\frac{p_1(\rho^\Gamma_{(-1)}(t))^2}{p_2(\rho^\Gamma_{(-1)}(t))} = \frac{\gamma^2 N_A^2}{8J^2} \ll 1,
\label{eq:D2perturb}
\end{eqnarray}
can be used to reveal the presence of entanglement at short times. We show in Fig.~\ref{fig:perturb} a numerical confirmation of these results for various values of $\gamma/J$ and $N=8$, which was obtained by simulating Eq.~\eqref{eq:me}. We note that, in the present context, utilizing symmetry-resolution is fundamental to detect entanglement: this is due to the fact that the negative eigenvalues in $\rho^\Gamma$ appear in sectors that are not macroscopically populated~\cite{Vi21}, so that moments without symmetry resolution would not be able to detect them.

\begin{figure}
    \includegraphics[width=\linewidth]{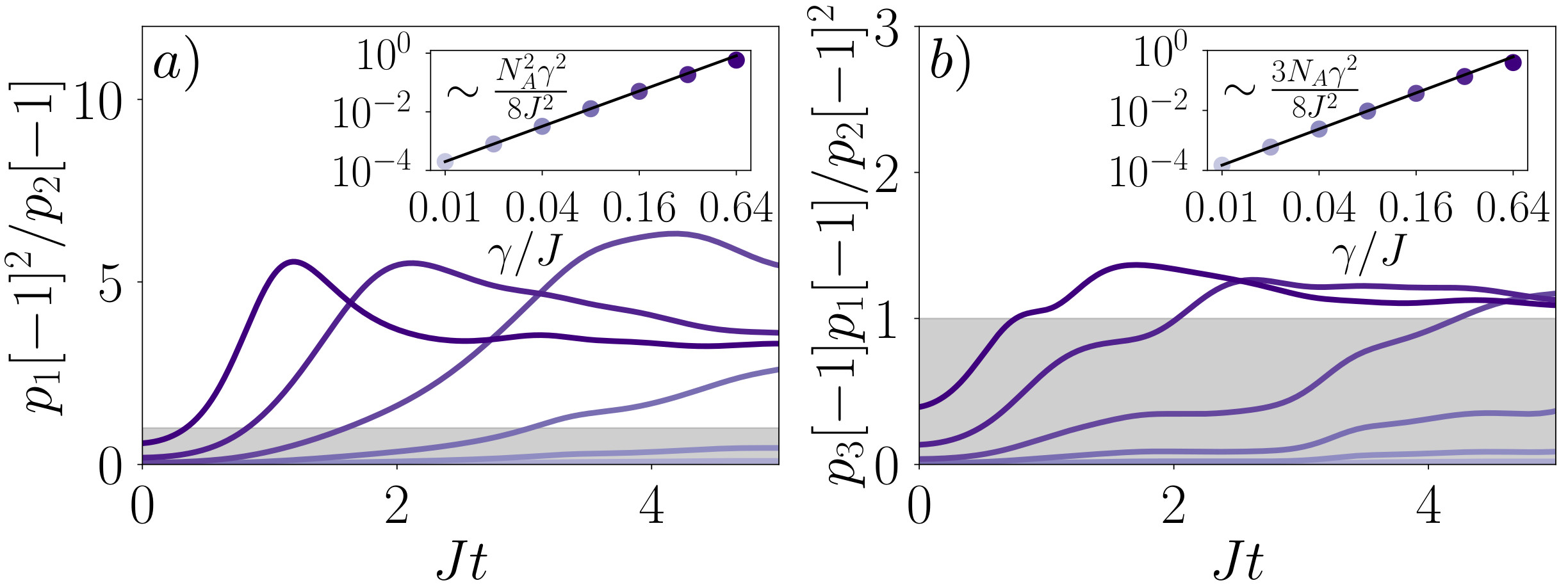}
    \caption{\textit{Symmetry resolved entanglement detection in quench dynamics with spin excitation loss.} We study SR-entanglement in quench dynamics in a system consisting of $N=8$ spins initialized in a N\'{e}el state $\ket{\downarrow\uparrow}^{\otimes N/2}$ and evolved with $H_{XX}$ subject to spin excitation loss with various rates $\gamma$ ($\gamma/J$ increases with the darkness of the color, see insets). We take $A=[1,2,3,4]$ and $B=[5,6,7,8]$. In panels a) and b), the $D_2$ ratio and $p_3$-PPT ratio of sector $ q=-1$ are shown, respectively. Entanglement is detected for values below unity in the shaded gray areas. The insets in a) and b) show the early time value at $t=0^+$ of the $D_2$-ratio a) and $p_3$-PPT ratio b), respectively, as function of the decoherence rate $\gamma/J$. Black lines are the perturbation theory results displayed in Eqs.~\eqref{eq:D2perturb}  and \eqref{eq:p3PPTperturb}.  }
    \label{fig:perturb}
\end{figure}

\subsection{Experimental demonstration in a trapped-ion quantum simulator}

In the previous section, we showed in an idealized theoretical setting that entanglement is generated --and can be revealed via SR-entanglement conditions-- at early times after a quantum quench. Here, we demonstrate this effect experimentally via the measurement  of the SR-$D_2$ and SR-$p_3$-ppT condition  using randomized measurement data taken at early times after quantum quench in a trapped ion quantum simulator (c.f. Ref.~\cite{Brydges2019}).  In particular, we show that the SR-$D_2$ condition and SR-$p_3$-PPT condition allow for a  fine-grained detection of bipartite entanglement, in regimes where the  corresponding global conditions \cite{EKH20} and conditions relying on the purities of different subsystems \cite{Brydges2019} are not conclusive.

In the experiment reported in Ref.~\cite{Brydges2019}, a one-dimensional spin-$1/2$-chain, consisting of $N=10$ spins, was initialized in  the N\'eel state $\ket{\!\!\uparrow\downarrow}^{\otimes5}$ and time-evolved with the Hamiltonian $H_{XY}$ [Eq.~\eqref{eq:Hxy}]
where the coupling parameter $J_{ij}$ follows the approximate power-law decay $J_{ij} \approx J_0/|i-j|^\alpha$, with $\alpha \approx 1.24$, $J_0 = 420\textrm{s}^{-1}$. 
The Hamiltonian evolution exhibits a global $U(1)$-symmetry conserving the total magnetization of the system (i.e., $[H, \sum_i Z_i]=0$).  Symmetry-breaking terms (such as $\sigma_i^+\sigma_j^++\text{h.c.}$) are strongly suppressed due to a large effective magnetic field \cite{Brydges2019}. As detailed in Refs.~\cite{Brydges2019,Vi21} weak decoherence effects are present in the experiment,  including imperfect initial state preparation, local spin-flips and spontaneous emission during the dynamics, and measurement errors model as local depolarization. 
Note that coherent spin-flips do not preserve the global magnetization and block-diagonal form of the (reduced) density matrix. On the timescales accessed in the experiment, these effects are however very weak (causing in numerical simulations including the above decoherence model a purity mismatch of the order of $10^{-5}$ of the full $10$-spin  density matrix $\rho$ vs.~the projected one $\rho_Q=\sum_q Q_q \rho Q_q$ at $t=5$ms).

In Ref.~\cite{Brydges2019} randomized measurements were performed at various times ($t=0\text{ms},\dots,5\text{ms}$) after the quantum quench. As described in detail in Ref.~\cite{Vi21} (see also Sec.~\ref{sec:class_shadows} and Appendix~\ref{App.shadows}), we can use this data  to estimate SR-PT moments and the SR entanglement conditions via classical shadow formalism \cite{huang2020shadow}. In Fig.~\ref{fig:exp}, we present the SR  $D_2$  and $p_3$-PPT conditions in the different sectors, for a subsystem consisting of the neighbouring spins $A,B=[4,5],[6,7]$ and where the partial transpose is taken in the subsystem $A=[4,5]$. Similar to the results of the previous subsection, both conditions detect entanglement at short times in the sector $q=-1$.  The corresponding global conditions, in particular the global $p_3$-PPT condition,  do not reveal the presence of entanglement in this regime [see Fig.~\ref{fig:exp} b)]. 

The fact that the SR-$D_2$ condition can reveal the presence of entanglement is particularly interesting from an experimental point of view as it implies that entanglement can be detected from the estimation of only two moments of the partial transpose (in a sector). For the shadow estimation of $D_2(-1)$, our rigorous bound from Theorem~\ref{thm:r1} ensures that $~\sim 1.3 \times 10^6$ measurements would be sufficient to guarantee entanglement detection with a probability of $95\%$. While this represents an upper bound, valid irrespective of the quantum state in question, for the specific states in the experiment only $8 \times 10^5$ have been performed. 
The errorbars of the experimental are then drawn at $1.96\sigma$ where the standard error of the mean $\sigma$ has been estimated for each data point using Jackknife resampling \footnote{For normally distributed data with empirical mean $\mu$, $\mu \pm 1.96\sigma$  defines a $95\%$ confidence interval. While normal distribution is here not guaranteed a priori, we checked through additional numerical simulations of many experiments (with fixed number of runs per experiment) that errorbars of length $1.96\sigma$ indeed approximate a confidence interval with confidence level $ 95\%$.}.

While the SR-$D_2$ condition requires only the estimation of first and second PT-moment, the third order SR-$p_3$-PPT condition [panel b)], allows to detect entanglement in an even wider time window. In comparison to the global $p_3$-PPT condition [red curve in panel b)], this clearly  demonstrates the benefit of taking symmetry-resolution into account.

\begin{figure}
    \centering
    \includegraphics[width=\linewidth]{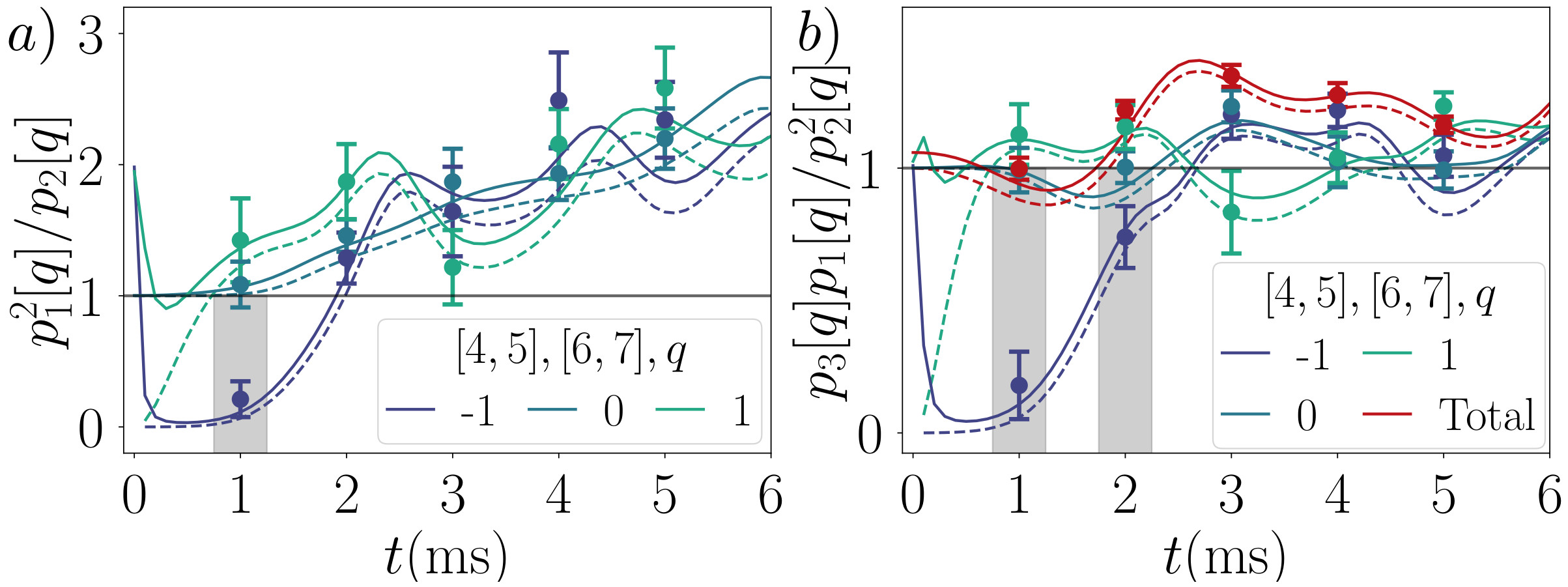}
    \caption{\textit{Experimental SR-entanglement detection in a trapped ion quantum simulator using data obtained in Ref.~\cite{Brydges2019}.} For a total system of $N=10$ spins and  subsystem $A, B=[4,5], [6,7]$,   we present in a) the SR-$D_2$ ratio and b) SR-$p_3$-PPT ratio as a function of time for various symmetry sectors . In both panels, entanglement is detected in regimes where the corresponding global conditions do not reveal entanglement, as indicated in the shaded grey areas (values below unity). Dashed (solid) lines are theoretical simulations of unitary dynamics (taking  decoherence effects into account), as detailed in Refs.~\cite{Brydges2019,Vi21}.}
    \label{fig:exp}
\end{figure}

\subsection{Entanglement detection in the ground state of the XXZ model}
The XXZ spin chain is a generalization of the Heisenberg chain including an anisotropy in the interaction along the $z$ direction, whose Hamiltonian reads:
\begin{equation}\label{eq:XXZ}
\begin{aligned}
    H=&-J\left( \sum_i X_i X_{i+1} + \sum_i Y_i Y_{i+1}+J_z \sum_i Z_i Z_{i+1}\right).
\end{aligned}
\end{equation}
We will fix $J=1$ as energy unit:  $J_z$ sets the strength of the anisotropy along the $z$-axis.
The phase diagram at zero temperature is known~\cite{Gogolin2004}: the system hosts an antiferromagnetic phase when $J_z<-1$, a Luttinger liquid for $J_z\in[-1,1]$, and a ferromagnetic one for $J_z>1$.
We might expect that the entanglement conditions we described in the previous sections will detect that the state is not PPT in the range $J_z\in]-\infty,-1]$. Since the XXZ spin chain exhibits a $U(1)$ symmetry related to magnetization conservation, we can exploit the symmetry-resolved counterpart of the $D_k$ conditions, the $p_3$-PPT and their optimized version $D_3^{\textrm{opt}}$.
\begin{figure}
\begin{minipage}{0.48\linewidth}
    \centering
    \includegraphics[width=\linewidth]{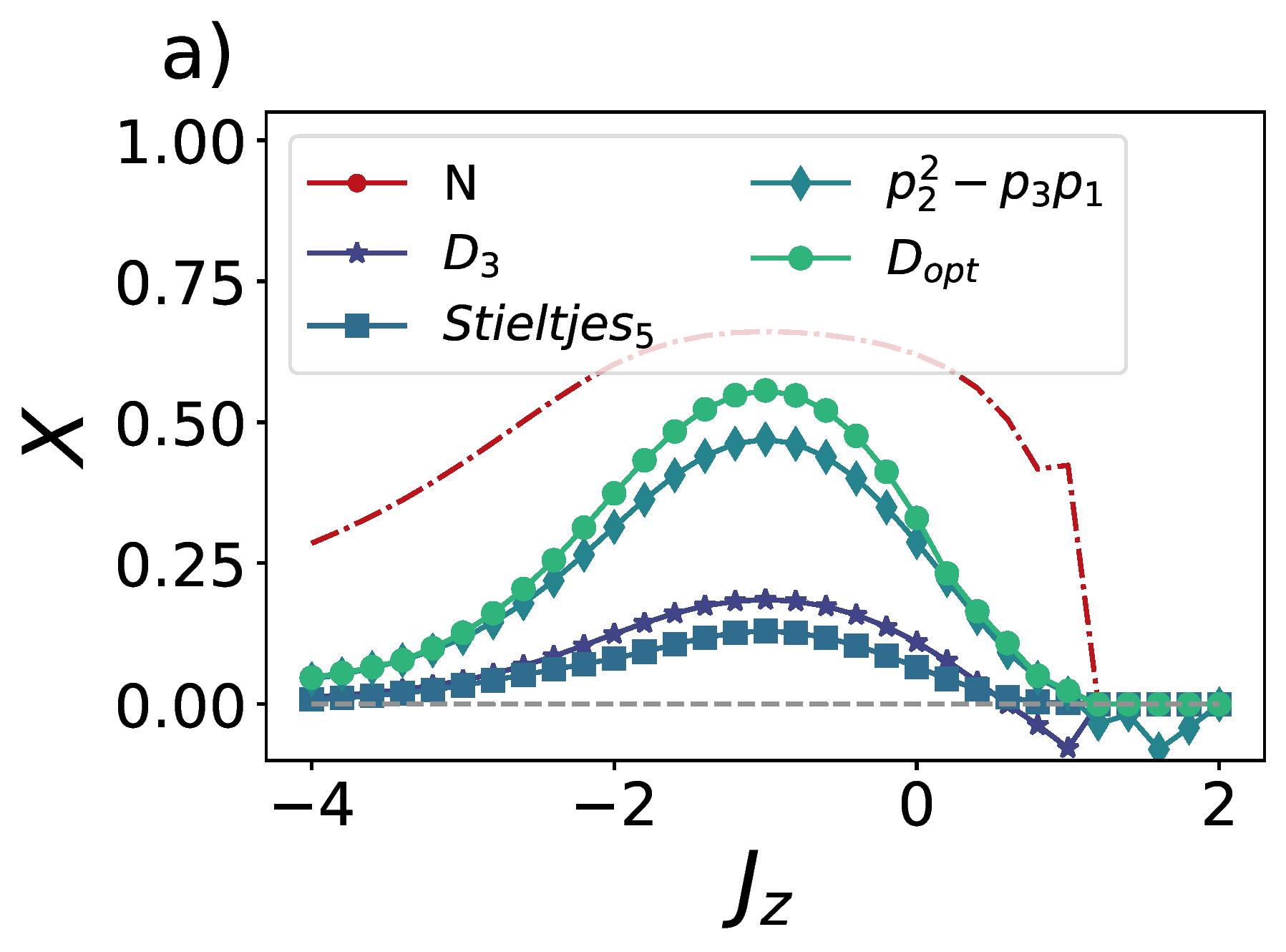}
    \end{minipage}
\begin{minipage}{0.48\linewidth}
    \centering
    \includegraphics[width=\linewidth]{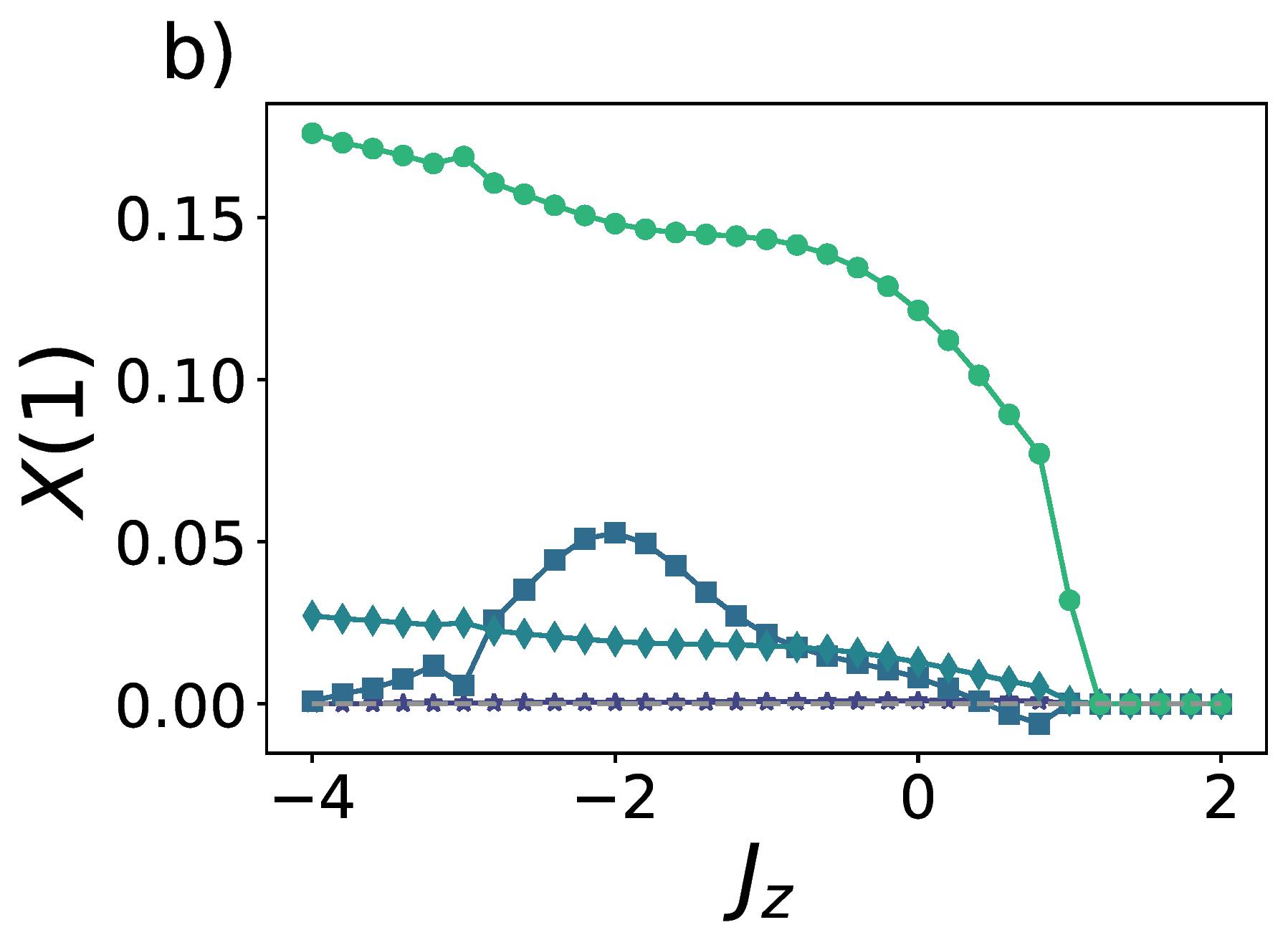}
    \end{minipage}
    \begin{minipage}{0.48\linewidth}
    \centering
    \includegraphics[width=\linewidth]{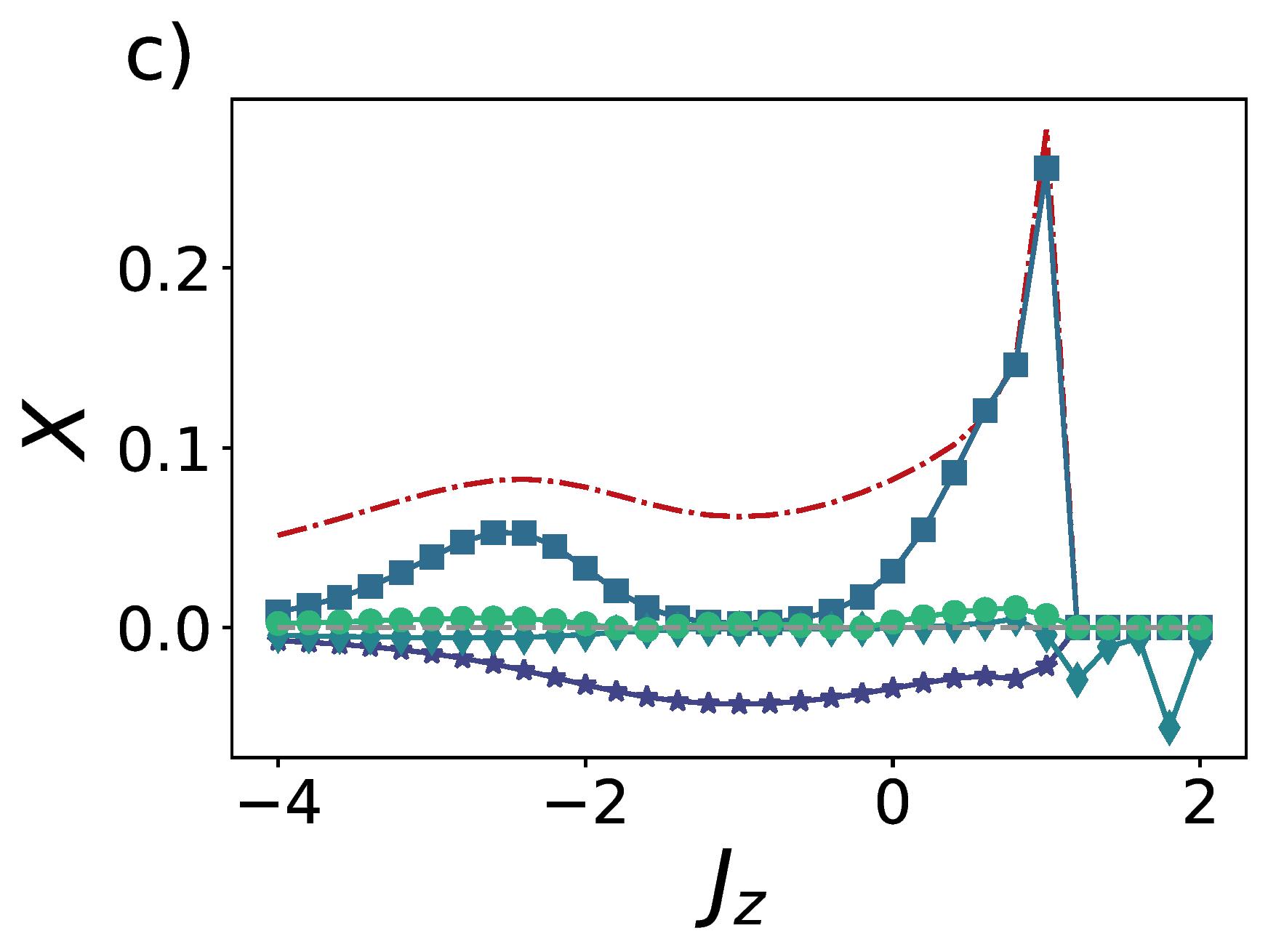}
    \end{minipage}
\begin{minipage}{0.48\linewidth}
    \centering
    \includegraphics[width=\linewidth]{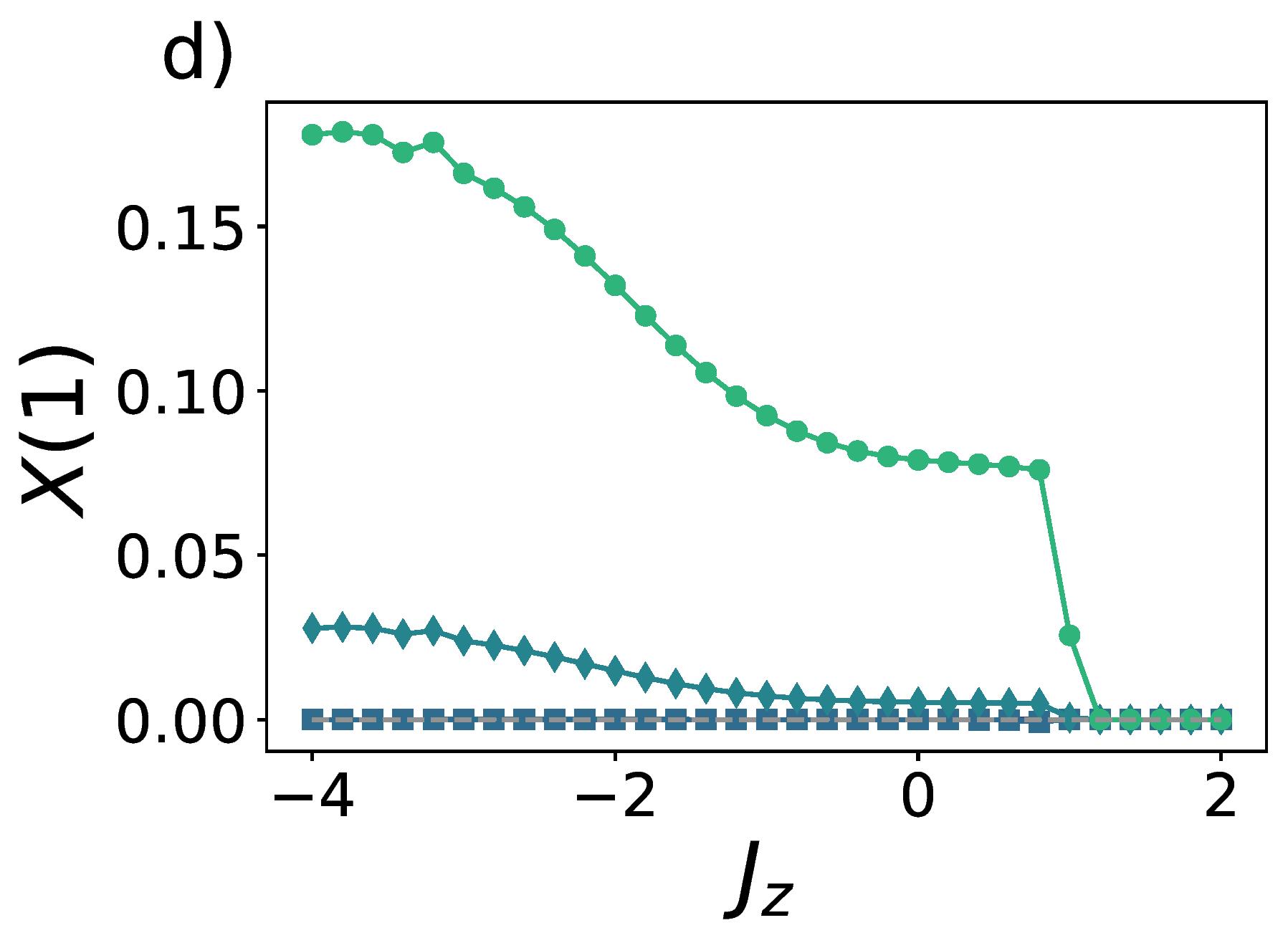}
    \end{minipage}
    \caption{Entanglement conditions on the ground state of XXZ model. With $X\in[D_3,D_{\textrm{opt}},Stieltjes_5,p_2-p_2^2,N]$ we denote the conditions computed on $\rho_A$ and with $X(1)$ the ones on $\rho_A(q=1)$ . Chain length $L=14$, subsystem length $\ell=10$. We consider in a)-b) a connected subsystem $A$ of length $\ell$ at the center of the chain; in c)-d) a disjoint interval $A$ consisting of $\ell/2$ sites at the beginning and $\ell/2$ sites at the end of the chain. We set the sign of the inequalities such that a positive value indicates the violation of a PPT condition, and thus the presence of entanglement. To compare data of different magnitude we multiply the $Stieltjes_5$ condition by $10^2$ in a), $10^5$ in b), $10^4$ in c), $10^5$ in d).}
    \label{fig:XXZ_sim}
\end{figure}

The simulation results are shown in Fig.~\ref{fig:XXZ_sim}. We consider the ground state of an open chain of $L=14$ sites. In a) and b), we select $\ell=10$ sites in the middle as subsystem $A$  and divide it in two parts $A=A_1\cup A_2$. We use the negativity as a reference to benchmark the efficiency of some entanglement conditions to detect entanglement between $A_1$ and $A_2$. 

In Fig.~\ref{fig:XXZ_sim}a), we calculate the $p_3$-PPT, the $D_3$, the optimal $D_3^{\textrm{opt}}$ condition and the Stieltjes condition using moments up to order five (see Appendix~\ref{App.Sti}). The convention we choose in the plot is that entanglement is detected whenever the value is positive. All the conditions work in most of the interval $J_z \in[-4,1]$, where we expect entanglement to be sizeable, except for $D_3$ failing in the vicinity of $J_z=1$. The presence of entanglement is confirmed by the calculation of the negativity (red line).

In Fig.~\ref{fig:XXZ_sim}b) we focus on the $q=1$ sector. In this case, we observe that all conditions indicate the presence of at least a negative eigenvalue in the sector $\rho^{\Gamma}_A(q=1)$ - that is, they are informative about which sector for the reduced density matrix contributes to violating PPT. In this specific instance, SR is however not fundamental in detecting entanglement beyond what non-SR conditions can. 

In Fig.~\ref{fig:XXZ_sim}c) and d), we carry out the same analysis for disconnected partitions. We consider $A=A_1\cup A_2$, where $A_1$ consists of the first $l/2$ sites and $A_2$ of the last  $l/2$, and $L=14, l=10$. In Fig.~\ref{fig:XXZ_sim}c) for $J_z \sim -1.9$ all the quantities except $Stieltjes_5$ are below zero, thus not revealing entanglement even though the negativity is positive. In this plot, one can also see that, for $J_z<-2$, the optimized condition $D_3^{\textrm{opt}}$ detects entanglement whereas both $p_3$-PPT and $D_3$ fail. This illustrates that the slight improvement obtained from the optimization in Sec.~\ref{sec.optimization} (see Fig.~\ref{fig.comparing}) can be decisive to detect the entanglement of physically relevant states from the first three moments only.

\subsection{Entanglement detection under constrained dynamics}
As a third example, we study the detection of mixed-state entanglement in subsystems of constrained spin models after a global quantum quench. Such models have been realized experimentally with neutral atoms in optical tweezer arrays coupled to Rydberg states~\cite{bernien51, ebadi2020quantum}. Below we simulate an experiment, in which moments of the partially transposed density matrix are obtained from a classical shadow involving global random unitaries available in current experimental setups. In particular we demonstrate that periodic revivals of mixed state entanglement can be detected from the conditions $D_3$ and $D_4$ (Eqs.~\eqref{eq:D3} and \eqref{eq:D4}) requiring only a small number of experimental runs.

We consider the Fibonacci chain with open boundary conditions described by the Hamiltonian 
\begin{align} \label{eq:pxp}
	H = \Omega \sum_i \mathcal{P}_{i-1} X_i \mathcal{P}_{i + 1}\,,
\end{align}
where $\mathcal{P}_i = \ket{0}_i\bra{0}$ are local projectors. As can be seen from~\eqref{eq:pxp}, each spin undergoes independent Rabi-oscillations as long as the neighbouring spins are in their ground state $\ket{0}$. This constraint breaks the tensor product structure of the Hilbert space (as it the case in a lattice gauge theory~\cite{Surace2019}). The model effectively resembles the experimental situation in~\cite{bernien51} if the Rydberg atoms are driven close to resonance and neighbouring atoms cannot be simultaneously in the state $\ket{1}$ due to the Rydberg blockade mechanism. The Hamiltonian~\eqref{eq:pxp} has recently attracted great interest in the context of quantum many-body scarring~\cite{turner2018, PhysRevLett.122.040603}. In particular, performing a quantum quench on special unentangled product states results in long-lived periodic revivals which have been attributed to the existence of quantum scarred eigenstates in the many-body spectrum~\cite{turner2018}. 

\begin{figure}
  \centering
  \includegraphics[width=\columnwidth]{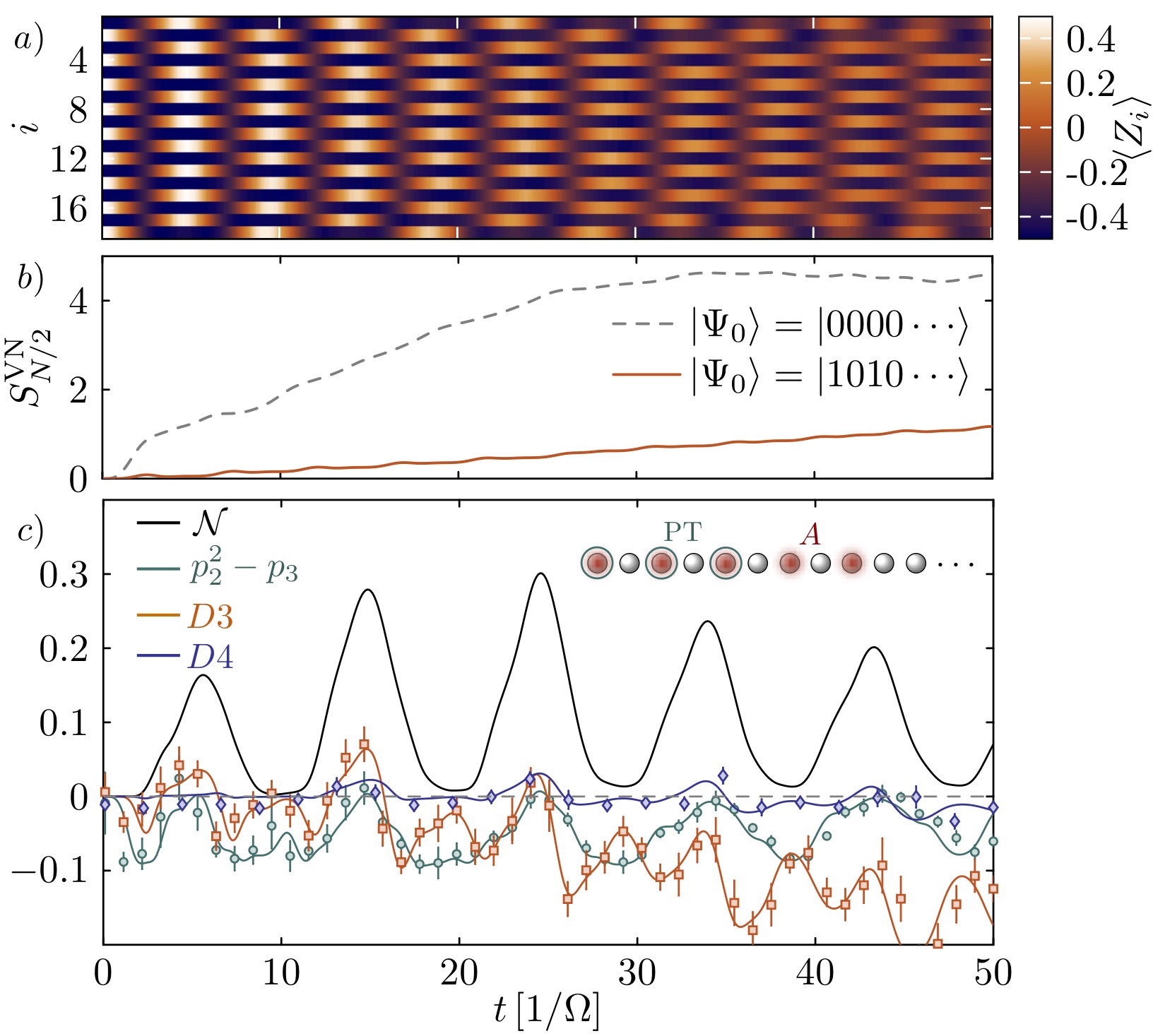}
  \caption{\textit{Entanglement detection in quench dynamics of a kinetically constrained Rydberg chain (\ref{eq:pxp}).}   a) Coherent oscillations of the $Z_i$-expectation values in a quench with a 18-site Fibonacci chain from a staggered initial state: $\ket{\Psi_0} = \ket{1010\dots}$. b) Von Neumann entropy of half partition of the chain as a function of time if a global quench is performed on 2 different initial states. c) Entanglement detection in a subsystem of the chain as indicated in the inset (positive values indicate entanglement). The revivals in the negativity indicate periodic entangling and disentangling of spins within the subsystem. The goblal $p_3-$PPT condition is unable to detect entanglement in the entire time window in contrast to the $D_3$ and $D_4$ conditions. The points are obtained from a classical shadow consisting of 5000 global random unitaries. Error bars are obtained by repeating the procedure 20 times and estimating the standard error.}
  \label{fig:rydberg}
\end{figure}

In the following we study the conditions given in Eqs.~\eqref{eq:D3} and~\eqref{eq:D4}, when a quench is performed from a product state that leads to kinetically constrained dynamics. To this end, the initial state $\ket{\Psi_0} = \ket{10}^{\otimes N/2}$ is time evolved with the Hamiltonian (\ref{eq:pxp}) up to $t = 50 / \Omega$. Fig.~\ref{fig:rydberg} a) shows the local $Z_i$-expectation values exhibiting long-lived persistent oscillations. This striking departure from a thermalizing behaviour is also reflected in the slow growth of entanglement entropy (Panel b). We now analyse the time resolved behaviour of mixed state entanglement for a subsystem depicted in the inset of Fig.~\ref{fig:rydberg} c). The revivals in the negativity indicate that spins in the subsystem get periodically entangled and disentangled with each other. Interestingly, the $p_3$-PPT condition is unable to detect the revivals, while $D_3$ yields positive values at the first 3 peaks in the negativity. At later times, the $D_3$ fails to detect the entanglement present in the system, but this entanglement is still captured by $D_4$.

Finally, we investigate the required number of experimental runs in order to measure the conditions up to given error bar. The classical shadow is constructed by sampling bit strings from the quantum state after applying a global random unitary on the subsystem $A$. At each point in Fig.~\ref{fig:rydberg} c), we collect 5000 bit strings in different random basis. Note that global random unitaries in Rydberg systems can be implemented via random quenches with local disorder potentials \cite{PhysRevA.97.023604}. The entire estimation of the conditions is repeated 20 times in order to obtain statistical uncertainties. Note that statistical covariances among the measured moments ${\rm tr}(\rho^{\Gamma})^n$ can give rise to nonuniform sizes for the error bars as can bee seen in~Fig.~\ref{fig:rydberg} c). In Fig.~\ref{fig:rydberg} c) we depict the $2 \sigma$ error bars, showing that entanglement can be detected with a moderate experimental effort.

\section{Conclusions and outlook}
\label{sec.conclusion}

The study of entanglement has a long and prominent history in a variety of disciplines. And with the advent of serious quantum technologies, reliable entanglement generation is more important than ever.
This work provides a novel and principled approach to reliably detect bipartite entanglement between subsystem $A$ and subsystem $B$.
We have presented a set of inequality conditions $D_k$ ($1 \leq k \leq 2^{|AB|}$).
Each $D_k$ is an inequality that compares the first $k$ moments of the partially transposed density operator. Violation of a single inequality implies that the underlying density operator cannot have a PSD partial transpose. This in turn implies that the state must be entangled. 
Conversely, if the underlying state is not PSD, then, there must exist at least one $D_k$ that is violated. 
This motivates a sequence of one-sided entanglement tests. Start with $D_3$ -- the lowest non-trivial condition -- and check whether it is violated. If this is the case, we are done. If not, we move on to the next higher condition ($D_4$) and repeat until we find a violation.
For states having an extensive conserved quantity (such as total magnetization, in the case of spin systems), both the density matrix and its partial transpose have a block-diagonal structure~\cite{CGS18}. In this case, it is advisable to apply these conditions directly to individual symmetry sectors of the partial transpose. 
The resulting sequence of symmetry-resolved conditions is stronger in the sense that lower moments (of blocks of the partial transpose) suffice to detect entanglement.
Importantly, this approach is not only conceptually sound, but also tractable from an experimental perspective. The classical shadows formalism \cite{huang2020shadow} allows for reliably estimating moments of the partial transpose from randomized single-qubit measurements. 
We demonstrated how to include the experimentally relevant situation of non identical (however independent) copies in the analysis and derived error bounds and confidence intervals for $D_2$, with a natural extension to quantities involving higher moments. Empirical evaluations complement our theoretical findings. Applications to several theoretical models, as well as experimental data, demonstrate both tractability and viability of our approach.

We are confident that this work opens up several interesting future research directions.
Firstly, the sequence of $D_k$'s is designed to detect bipartite entanglement in a reliable and experimentally accessible fashion. A natural next step is to try to extend similar ideas to multipartite entanglement detection, e.g. using non--linear entanglement witnesses ~\cite{GuehneLuetkenhaus2006,Jungnitsch2011}.
Secondly, the complete sequence of $D_k$'s is used to answer a binary question: is the partial transpose negative or not? Entanglement measures, like the negativity, address entanglement in a quantitative fashion, but are also harder to estimate. Is it possible to use moments (or other density matrix functionals) to define entanglement measures that are experimentally tractable?
The statistical analysis of the estimation procedure is also far from complete. We have shown that
independence between the states that are produced in each iteration of an experiment is enough to derive statistically sound confidence intervals for estimating matrix moments with classical shadows. This addresses the practically relevant case of drifting sources, but further extensions to correlated states would also be appealing. In this context the quantum de Finetti theorem seems highly relevant.
In future work, we will also investigate how importance sampling~\cite{hadfield2020measurements,Rath2021} and/or derandomization~\cite{huang2020shadow,huang2021derandomization} can further improve moment estimation based on classical shadows.
Finally, another promising direction of research would be 
to try to detect and characterize phase transitions in quantum mechanical Hamiltonians at finite (non-zero) temperature. Quantum phase transitions at zero temperature 
originate from quantum fluctuations, whereas quantum phase transitions at finite temperature are due to thermal fluctuations. Following Ref.~\cite{Lu20}, quantum phase transitions at finite temperature can be studied using the negativity. It would be interesting to investigate whether low-order PT moments, intimately related to the negativity, can also be used to this end.

\paragraph*{Note:} While completing the writing of the present work, we became aware of a work by Yu \textit{et al.}~\cite{Yu}, in which similar questions have been addressed.

\section*{Acknowledgments}

We would like to thank Ion Nechita for pointing out the Stieltjes moment problem to us. J. C., B. K. and A. N. acknowledge financial support from the Austrian Science Fund (FWF) stand alone project: P32273-N27, the FWF: FG-5
and the SFB BeyondC. BV acknowledges funding from the Austrian Science Fundation (FWF, P 32597 N), and the French National Research Agency (ANR-20-CE47-0005, JCJC project QRand).  The work of MD and VV is partly supported by the ERC under grant number 758329 (AGEnTh), by the MIUR Programme FARE (MEPH), and has received funding from the European Union's Horizon 2020 research and innovation programme under grant agreement No 817482 (Pasquans). P.C.~acknowledges support from ERC under Consolidator grant number 771536 (NEMO). C.K., A.E.~and P.Z.~acknowledge support by 
European Union's Horizon 2020 research and innovation programme under Grant Agreement No. 817482 (Pasquans) and Simons Collaboration on Ultra-Quantum Matter, which is a grant from the Simons Foundation (651440, P.Z.).

\appendix
\section{Appendix 1. Descartes' rule of signs}\label{App.Des}
Let $A$ be a Hermitian matrix of dimension $d$. Its eigenvalues $\lambda_1,\dots,\lambda_d$ are the roots of the characteristic polynomial
\[
    P(t) = \det \left( A - t \, \mathds{1} \right) = \prod_{i=1}^d (\lambda_i-t).
\]
For convenience, let us now consider the polynomial $P(-t)$, which effectively replaces the positive eigenvalues of $A$ by negative ones and vice versa. The coefficients of this polynomial can be expressed using the elementary symmetric polynomials in its roots, $e_i(\lambda_1,\dots,\lambda_d)$, defined as
\[
    e_i(\lambda_1,\dots,\lambda_d) = \sum_{1\leq j_1 < \cdots < j_i \leq d} \lambda_{j_1} \dots \lambda_{j_i},
\]
for $i=1,\dots,d$ and with $e_0(\lambda_1,\dots,\lambda_d)=1$. This yields 
\[
    P(-t) = \sum_{i=0}^d e_i(\lambda_1,\dots,\lambda_d) \, t^{d-i}.
\]
For a polynomial with real roots (as it is the case here), Descartes' rule of sign states that the number of positive roots is given by the number of sign changes between consecutive elements in the ordered list of its non-zero coefficients (see Ref.~\cite{Be16} and references therein). The matrix $A$ is PSD iff $P(-t)$ has only negative roots, which by Descartes' rule is the case iff there is no sign change in the ordered list of its non-zero coefficients, i.e., iff $e_i(\lambda_1,\dots,\lambda_d) \geq 0$ for all $i=1,\dots,d$, since $e_0(\lambda_1,\dots,\lambda_d) =1$.

\section{Appendix 2. Stieltjes moment problem}\label{App.Sti}
Given a sequence of moments, $(m_n)_{n=0}^d$, the (truncated) Stieltjes moment problem consists in finding necessary and sufficient conditions for the existence~\footnote{If such a measure exists, one may wonder whether it is unique or not. For our purposes, it will be enough to discuss only its existence.} of a measure $\mu$ on the half-line $[0,\infty)$ such that
\begin{equation}\label{eq:moments}
m_n=\int_0^\infty x^n{\rm d}\mu(x), \, \forall  n\in \{0,\dots,d\}.
\end{equation}
 Defining the matrices
\begin{equation}
A(n)=
\left( \begin{array}{ccccc}
m_0 & m_1 & m_2 & \cdots & m_n\\
m_1 & m_2 & m_3 & \cdots & m_{n+1}\\
m_2 & m_3 & m_4 & \cdots & m_{n+2}\\
\vdots & \vdots & \vdots & \ddots & \vdots\\
m_n & m_{n+1} & m_{n+2} & \cdots & m_{2n}
\end{array} \right)
\end{equation}
and
\begin{equation}
B(n)=
\left( \begin{array}{ccccc}
m_1 & m_2 & m_3 & \cdots & m_{n+1}\\
m_2 & m_3 & m_4 & \cdots & m_{n+2}\\
m_3 & m_4 & m_5 & \cdots & m_{n+3}\\
\vdots & \vdots & \vdots & \ddots & \vdots\\
m_{n+1} & m_{n+2} & m_{n+3} & \cdots & m_{2n+1}
\end{array} \right),
\end{equation}
a solution to this problem can be stated as follows~\cite{Stieltjes}. If $d$ is odd -- such that $d=2k+1$ for some integer $k$ -- there exists such a measure $\mu$ if and only if
\begin{equation}
    A(k)\geq0, \; B(k)\geq0 \textrm{ and } (m_k,\dots,m_{2k+1})^T \in \mathcal{R}[ A(k)], \label{eq:cond_odd}
\end{equation}
where, given a matrix $M$, the notation $M\geq0$ indicates that $M$ is PSD and $\mathcal{R}(M)$ denotes the range of $M$.
If $d$ is even -- such that $d=2k$ for some integer $k$ -- there exists such a measure $\mu$ if and only if 
\begin{multline}
    A(k)\geq0, \; B(k-1)\geq0 \textrm{ and } \\ (m_{k+1},\dots,m_{2k})^T \in \mathcal{R}[ B(k-1)]. \label{eq:cond_even}
\end{multline}

These solutions to the Stiltjes moment problem can be used to obtain entanglement conditions. 
Given $\lambda_1,\ldots,\lambda_r$ the eigenvalues of $\rho^\Ga$ for some density matrix $\rho$, let us define the (atomic) eigenvalue distribution function
\[
{\rm d}\mu(x)=\sum_{i=1}^r\delta(x-\lambda_i)\,,
\]
where $\delta$ is the Dirac delta distribution. If $\rho$ is PPT, this density function has support on $[0,\infty)$ and reproduces the moments of $\rho^\Ga$, as
\[
m_n=\int_0^\infty x^n{\rm d}\mu(x) = p_n(\rho^\Ga).
\]
Therefore, according to the solution of the Stieltjes moment problem mentioned above, the moments of any PPT state necessary satisfy either condition~\eqref{eq:cond_odd} or~\eqref{eq:cond_even}, depending on the value of $r$. The violation of any of these conditions for a set of PT moments thus reveals that the corresponding state must be entangled. The range condition may require the knowledge of all the moments to be checked, but the PSD conditions can be broken into sets of simpler conditions. And some of them only involve low order moments. Indeed, it is well known (see e.g.~\cite{Sylvester}) that a matrix is PSD if and only if all its principal minors are non-negative. For instance, looking at the principal minor at the intersection of the first two rows and columns of $B(k)$, one obtains the condition $m_1 m_3 - (m_2)^2 \geq0$. This condition is nothing but the $p_3$-PPT condition (which we know is useful to detect entanglement~\cite{EKH20}). Extending this principal minor to the third row and column, one gets another PPT condition: 
\begin{equation}
\det
\left( \begin{array}{ccc}
m_1 & m_2 & m_3 \\
m_2 & m_3 & m_4 \\
m_3 & m_4 & m_5 
\end{array} 
\right)
\geq 0.
\end{equation}
We call this condition $\textrm{Stieltjes}_5$.
We illustrate in Sec.~\ref{sec.Applications} (c.f. Fig~\ref{fig:XXZ_sim}) that this condition is also useful for entanglement.
Numerical computations suggest that this condition is a powerful tool to detect the entanglement of random mixed states, in the sense that it detects more random entangled states than either $p_3^{\textrm{opt}}$ or $D_5$. 
Not all Stieltjes moment conditions are this powerful, though.
For instance, the principal minor condition for the first two rows and columns of $A(k)$ is trivial.  

Note that, because we consider here an atomic density function, we have $m_0=r$. We could naturally renormalize the density function so that $m_0=1$, but it would imply a re-scaling of the first moment, i.e., the trace of $\rho^\Ga$, would be $1/r$. Since the partial transpose and the density function cannot be normalized at the same time, we chose to keep normalized partial transposes. 

\section{Appendix 3. Optimizing conditions involving moments up to degree three}\label{App.Lagrange}
Given a PSD matrix $A$, with non-zero eigenvalues $\lambda_1,\dots,\lambda_r$, for some $r\in [1,\dim A]$, consider the Lagrangian function
\begin{multline}
    L(\lambda_1,\dots,\lambda_r,C_1,C_2) = \\ \sum_{i=1}^r{\lambda_i^3} + C_1 \left( \sum_{i=1}^r{\lambda_i^2} - p_2 \right) + C_2 \left( \sum_{i=1}^r{\lambda_i} - p_1 \right),
    \label{Lagrangian}
\end{multline}
where $C_1$ and $C_2$ are Lagrange multipliers. 

Here we show that, for all $1\le r\le\dim A$, the stationary points $(\la_1,\ldots,\la_r)$ of the Lagrangian function~\eqref{Lagrangian} are such that the variables $\la_i$ can take at most two distinct values. These stationary points, for which the derivatives of the Lagrangian~\eqref{Lagrangian} with respect to each variable vanish, satisfy the set of equations
\begin{align}
    3 \la_i^2 + 2 C_1 \la_i + C_2 &= 0\,,\qquad i=1,\dots,r\label{eq:Lagrange1}\\
    \sum_{i=1}^r \la_i^2 &= p_2\,,\label{eq:Lagrange2}\\
    \sum_{i=1}^r \la_i &= p_1\,.\label{eq:Lagrange3}
\end{align}
We first sum up Eq.~\eqref{eq:Lagrange1} for all values of $i$ and then insert Eqs.\eqref{eq:Lagrange2} and~\eqref{eq:Lagrange3} into it. This yields
\[
    C_2 = \frac{-2 C_1 p_1 - 3 p_2}{r}.
\]
Inserting this relation into Eq.~\eqref{eq:Lagrange1} and considering this equation for two distinct values of $i$, say 1 and $k\neq 1$, one can eliminate the variable $C_1$ to get a relation between $\la_1$ and $\la_k$. After some algebra, one finds
\[
    \la_k = \la_1 \textrm{ or } \la_k = \frac{\la_1 p_1 - p_2}{\la_1 r -p_1}.
\]
Since this argument holds for any $k\neq 1$, it must hold that the eigenvalues $\la_i$ are either all equal or can only take two distinct values. In the first case, in which all the eigenvalues are equal, one obtains the isolated points $(p_2,p_3)=(1/r,1/r^2)$ in Fig.~\ref{fig.comparing} in the main text. In the second case, the rank $r$ PSD matrices corresponding to the stationary points of the Lagrangian~\eqref{Lagrangian} have a spectrum with $r_a$ degenerate eigenvalues $\lambda_a$ and $r-r_a$ eigenvalues $\la_b$. Assuming, without loss of generality, $\la_a>\la_b$, one can show that the minimal value of the third moment is obtained when $r_a=r-1$.  

\section{Appendix 4. Estimating block PT moments with classical shadows: rigorous confidence regions}\label{App.shadows}

Classical shadows are a convenient formalism to reason about predicting properties of a quantum system based on randomized measurements \cite{huang2020shadow,paini2019approximate}.
In Ref.~\cite{EKH20}, a classical shadow error analysis was carried out for the estimation of partial transpose moments (of order two and three) from randomized single-qubit measurements performed on a single copy of the experimental state at a time. In this extended appendix, we go one step further an present a detailed error analysis in case one does not assume that the states produced in the experiment are identical in each iteration of the experiment (as it would be the case for a ``drifting'' source). For completeness, we first recall the fundamental ideas of classical shadows. After this presentation, we focus on the estimation of quadratic observables and provide rigorous error bounds for the estimation of quadratic polynomials in the moments of the partial transpose of a projected density operator. 

\subsection{Classical shadows}

Suppose, we are interested in quantum systems comprised of $n$ qubits. 
Suppose furthermore, that we can perform certain unitary transformations $U \in \mathcal{E}$ (ensemble), as well as a measurement in the computational basis: $\left\{|b \rangle \! \langle b|:\; b \in  \left\{0,1\right\}^n \right\}$. 
It is instructive to analyze the quantum-to-classical channel that arises from first performing a randomly selected unitary transformation $\rho \mapsto U \rho U^\dagger$ followed by a computational basis measurement:
\begin{align}
\mathcal{M}_{\mathcal{E}}(\rho) =& \int_{\mathcal{E}} \sum_{b \in \left\{0,1\right\}^{n}} 
\mathrm{Pr} \big[ \hat{b}=b| U \rho U^\dagger \big] U^\dagger |b \rangle \! \langle b| U \mathrm{d}U \nonumber \\
=& \int_{\mathcal{E}} \mathrm{d}U \sum_{b \in \left\{0,1\right\}^n} \langle b| U \rho U^\dagger |b \rangle U^\dagger |b \rangle \! \langle b| U. \label{eq:M}
\end{align}
The integral over $\mathcal{E}$ is an average over all possible classically randomized measurement settings. The summation over $b$ averages over quantum randomness associated with measurement outcomes (Born's rule). It is easy to check that $\mathcal{M}_{\mathcal{E}}(\cdot)$ is always a quantum channel, i.e.\ a completely positive and trace-preserving map. 

Viewed as a linear operator, this channel is also invertible if the underlying ensemble $\mathcal{E}$ is sufficiently expressive. 
More precisely, we require that the complete family $\left\{ U^\dagger |b \rangle \! \langle b| U:\; U \in \mathcal{E} \right\}$ of admissible basis measurements is \emph{tomographically complete}. That is, for all $\rho \neq \sigma$, there exists a $U \in \mathcal{E}$ and a outcome $b \in \left\{0,1\right\}^n$ such that $\langle b| U \rho U^\dagger \rangle \neq \langle b| U \sigma U^\dagger |b \rangle$.

\begin{fact}
Suppose that $\left\{U^\dagger |b \rangle \! \langle b| U:\; U \in \mathcal{E},b \in \left\{0,1\right\}^n\right\}$ is a tomographically complete family of basis measurements. Then, $\mathcal{M}_{\mathcal{E}}$ has a well-defined and unique inverse $\mathcal{M}_{\mathcal{E}}^{-1}(\cdot)$.
\end{fact}

From now on, we will always assume that we are dealing with tomographically complete families of basis measurements.
Classical shadow estimation with randomized measurements is based on the following basic routine:
\begin{enumerate}
\item \emph{state preparation:} prepare a copy of the unknown quantum state $\rho$;
\item \emph{randomized single-shot measurement:} sample $U \sim \mathcal{E}$ at random, transform $\rho \mapsto U \rho U^\dagger$ and measure in the computational basis; \newline
\item \emph{construct a classical snapshots:} upon receiving outcome $\hat{b} \in \left\{0,1\right\}^n$, compute
\begin{align}
\hat{\rho} =& \mathcal{M}_{\mathcal{E}}^{-1} \left( U^\dagger |\hat{b} \rangle \! \langle \hat{b}| U \right).
\label{eq:classical-shadow}
\end{align}
\end{enumerate}
By construction, each snapshot is a random matrix that exactly reproduces the true underlying state $\rho$ in expectation (over both the classical choice of transformation and the quantum randomness in the basis outcome). That is,
\begin{align*}
\mathbb{E} \left[ \hat{\rho} \right] =& \mathcal{M}_{\mathcal{E}}^{-1} \left( \mathbb{E}_{U\sim \mathcal{E}} \mathbb{E}_{\hat{b}\in \left\{0,1\right\}^n} U^\dagger |\hat{b} \rangle \! \langle \hat{b}| U \right) \\
=& \mathcal{M}_{\mathcal{E}}^{-1} \left( \mathcal{M}_{\mathcal{E}} (\rho) \right) = \rho.
\end{align*}
Statistically speaking, $\hat{\rho}$ is an unbiased estimator of the underlying quantum state $\rho$. But a single snapshot is only a very poor estimator. This situation changes if we have access to multiple independent snapshots $\left\{\hat{\rho}_1,\ldots,\hat{\rho}_N \right\}$. We call such a collection a \emph{classical shadow} of $\rho$ with size $N$.
Forming the empirical average of snapshots within a classical shadow produces ever more accurate approximations of the true underlying state:
\begin{equation}
\tfrac{1}{N}\sum_{i=1}^N \hat{\rho}_i \overset{N \to \infty}{\longrightarrow} \tfrac{1}{N} \sum_{i=1}^N \mathbb{E} \left[ \hat{\rho}_i \right] = \rho. \label{eq:convergence}
\end{equation}
The main results in Ref.~\cite{huang2020shadow} highlight that classical shadows of moderate size already allow joint estimation of many interesting state properties.
More precisely, the classical shadow formalism allows for computing powerful a-priori bounds on the convergence behavior of such estimators. 

\subsection{Implicit assumptions and relaxations thereof}

Before moving on, it is worthwhile to delineate implicit assumptions within the classical shadows model. It turns out that almost all of them can be relaxed without threatening statistical guarantees like error bounds or confidence intervals.

\paragraph{Perfect/noiseless measurements:}
the original randomized measurement framework 
is contingent on perfect knowledge of the average quantum-to-classical channel~\eqref{eq:M}. Erroneous executions of the ensemble rotation $U$, or noisy executions of the subsequent computational basis measurement can thwart this assumption. However, recent results \cite{flammia2020shadow,koh2020shadow} highlight that a suitably extended shadow formalism can handle such imperfections. The key idea is to adjust the inversion formula~\eqref{eq:classical-shadow} appropriately. 
Ref.~\cite{koh2020shadow} achieves such an adjustment by assuming explicit knowledge of the (average) noise channel, while Ref.~\cite{flammia2020shadow} actually goes a step further and proposes a tractable calibration protocol that reveals sufficient information to appropriately correct Eq.~\eqref{eq:classical-shadow}.

\paragraph{Independent and identically distributed state copies:} because quantum measurements are typically destructive, most estimation protocols assume access to a perfect source that produces independent and identically distributed (\textit{iid}) copies of the underlying quantum state $\rho$. Formally, $N$ queries of such a perfect \textit{iid} state source produce the state $\rho^{\otimes N}$ and we subsequently proceed to measure independent copies sequentially. This particular tensor product structure is a strong assumption that combines stochastic independence (individual state copies are completely uncorrelated) with identical distribution (all state copies are identical). This second assumption is often violated in concrete experimental architectures. Small fluctuations within the device can lead to state copies that, although uncorrelated, vary in time (``drifting source''). $N$ state copies produced by such drifting (but independent) sources can be modelled by a sequence $\left\{\rho_i \right\}_{i=1}^N$ of non-identical quantum states.

The classical shadow formalism can readily handle drifting (but independent) sources. Each snapshot $\hat{\rho}_i$ will have a different expectation value and Eq.~\eqref{eq:convergence} needs to be adjusted accordingly:
\begin{equation}
\tfrac{1}{N} \sum_{i=1}^N \hat{\rho}_i \longrightarrow \tfrac{1}{N} \sum_{i=1}^N \mathbb{E} \hat{\rho}_i = \tfrac{1}{N}\sum_{i=1}^N \rho_i =: \rho_{\mathrm{avg}}.
\end{equation}
Hence, empirical averages of classical shadows are well-suited for approximating linear properties of the average source state $\rho_{\mathrm{avg}}$. We will show below that this desirable feature extends to the shadow estimation of polynomials as well. Note that it is, sufficient to know e.g. that the polynomial, $D_2$ evaluated at the average state is negative to ensure that the source is capable of producing entanglement, as in this case there must existe a $\rho_k$ which leads to a negative value.

In order to handle drifting sources, we will also assume access to trusted classical randomness that allows us to randomly select unitary transformations $U \in \mathcal{E}$. Mild by comparison, this assumption also features as an explicit (or implicit) assumption in other statistically sound treatments of entanglement detection, see e.g.\ \cite{dirkse2020entanglement} and references therein.

Based on these assumptions, we will establish rigorous error bounds and confidence intervals for predicting polynomial functions based on classical shadows. Rather than treating this problem in full generality, we focus on estimating 
\begin{equation*}
D_2^{(i)}(\rho)= p_1 (P_i \rho^\Ga P_i)^2 - p_2 (P_i \rho^\Ga P_i),
\end{equation*}
i.e., the first symmetry-resolved polynomial that is capable of detecting NPT entanglement. Theorem~\ref{thm:D2} below highlights that an order of $2^{|AB|} / \epsilon^2$ randomized measurements suffice to approximate $D_2^{(i)}(\rho_{\mathrm{avg}})$ up to additive accuracy $\epsilon$. Corollary~\ref{cor:D2} reformulates this insight in terms of confidence intervals.

Finally, we point out that assuming access to independent state copies (tensor product structure) is not a mild assumption. But, the classical shadow estimators below are invariant under permuting individual state copies. Permutation invariance suggests that strong proof techniques from quantum cryptography -- like the quantum de Finetti theorem, see e.g.\ \cite[Chapter 7]{watrous2018book} 
and references therein -- may be applicable and allow for relaxing the independence assumption as well. We intend to address this in future work.

\subsection{Predicting linear functions with classical shadows}

We are now ready to discuss the simplest use-case of classical shadows: estimate a linear function, say $\mathrm{tr}(O \rho)$, based on $N$ randomized measurements of \emph{independent} (but not necessarily identical) states. We can achieve this by simply replacing the unknown quantum state $\rho$  by an empirical average of snapshots within a classical shadow:
\begin{align}
\hat{o}_{(N)}=&\tfrac{1}{N} \sum_{i=1}^N \mathrm{tr}\left( O \hat{\rho}_i \right) \quad \text{obeys} \label{eq:linear-estimator}\\
\mathbb{E} \hat{o} =& \mathrm{tr} \left( O \tfrac{1}{N} \sum_{i=1}^N \mathbb{E} \hat{\rho}_i\right) = \mathrm{tr} \left( O \tfrac{1}{N} \sum_{i=1}^N \rho_i \right) = \mathrm{tr} \left( O \rho_{\mathrm{avg}}\right). \nonumber
\end{align}
Independence ensures that the individual snapshots $\hat{\rho}_i$ are stochastically independent random matrices. This in turn implies that each $\mathrm{tr}(O \hat{\rho}_i)$ is a stochastically independent random variable. Empirical averages of independent random variables tend to concentrate sharply around their expectation value -- regardless of the underlying distribution. The \emph{variance} is an important summary parameter that can control the rate of convergence. Chebyshev's inequality, for instance, implies for $\epsilon >0$
\begin{align}
& \mathrm{Pr} \left[ \left| \hat{o}_{(N)}-\mathrm{tr} \left( O \rho_{\mathrm{avg}}\right) \right| \geq \epsilon \right]
\leq \tfrac{1}{\epsilon^2} \mathrm{Var} \left[ \hat{o}_{(N)}\right] \nonumber \\
=& \tfrac{1}{N \varepsilon^2} \left( \tfrac{1}{N}\sum_{i=1}^N \mathrm{Var} \left[ \mathrm{tr} \left( O \hat{\rho}_i \right) \right] \right).
\label{eq:chebyshev}
\end{align}
This tail bound is a consequence of independence alone. The remaining average variance does depend on the ensemble $\mathcal{E}$. Different ensemble give rise to different variance contributions \cite{huang2020shadow,paini2019approximate}. 
Here, we focus on the practically most relevant case: randomized, \emph{single-qubit measurements}. 
Each qubit is measured in either the $X$-, the $Y$-, or the $Z$-basis. 
More formally, the ensemble $\mathcal{E}$ subsumes random single qubit Clifford rotations. That is $U = U_1 \otimes \cdots \otimes U_n$ with $U_1,\ldots,U_n \overset{\textit{iid}}{\sim} \mathrm{Cl}(2)$
and $\mathrm{Cl}(2)$ denotes the single-qubit Clifford group, i.e., the finite group generated by Hadamard and phase gates. More generic single-qubit ensembles (like Haar-random unitaries) are also an option -- what matters is that the single qubit ensemble forms a 3-design \cite{dankert2009unitary_designs,gross2007evenly}. The (single- and multi-qubit) Clifford group is one ensemble with this feature \cite{zhu2017clifford,webb2016clifford,kueng2015stabilizer}.
As demonstrated in Ref.~\cite{huang2020shadow}, the 3-design assumption allows us to compute the measurement channel~\eqref{eq:M}, as well as its inverse.
Let $\mathcal{D}_{1/3}(X)=1/3\left(X+\mathrm{tr}(X)\mathbb{I}\right)$ denote the single-qubit depolarizing channel with parameter $1/3$. Then,
\begin{align*}
\mathcal{M}\left(\bigotimes_{k=1}^n X_k\right) =& 
\bigotimes_{k=1}^n \mathcal{D}_{1/3}(X_k) = 3^{-n} \bigotimes_{k=1}^3 \left( X_k + \mathrm{tr}(X_k) \mathbb{I} \right), \\
\mathcal{M}^{-1}\left( \bigotimes_{k=1}^n X_k \right) =& \bigotimes_{k=1}^n \mathcal{D}_{1/3}^{-1}(X_k)
= \bigotimes_{k=1}^n \left( 3 X_k - \mathrm{tr}(X_k) \mathbb{I} \right),
\end{align*}
and we refer to \cite[Supplementary Information, Section C]{huang2020shadow} for details.
We see that the measurement channel (and its inverse) factorize nicely into a tensor product of single-qubit operations.
This is also true for snapshots~\eqref{eq:classical-shadow} within a classical shadow:
\begin{align}
\hat{\rho}
= \bigotimes_{k=1}^n \left( 3U_k^\dagger |\hat{b}_k \rangle \! \langle \hat{b}_k |U_k - \mathbb{I} \right). \label{eq:pauli-shadow}
\end{align}
This explicit formulation allows for deriving powerful and useful variance bounds, see \cite[Supplementary Information, proof of Proposition~S3]{huang2020shadow}.

\begin{lemma}[linear variance bound] \label{lem:linear-variance}
Fix an observable $O$ and
suppose that $\hat{\rho}_i$ is the snapshot~\eqref{eq:pauli-shadow} of an unknown quantum state. 
Then,
\begin{equation}
\mathrm{Var} \left[ \mathrm{tr}(O \hat{\rho}_i) \right] \leq 2^{\mathrm{w}(O)} \mathrm{tr}(O^2),
\end{equation}
where 
$\mathrm{w}(O)$ denotes the observable's weight; that is, the number of qubits on which it acts nontrivally.
\end{lemma}

We can combine Lemma~\ref{lem:linear-variance} with Eq.~\eqref{eq:chebyshev} to obtain
\begin{equation}
\mathrm{Pr} \left[ \left| \hat{o}_{(N)}-\mathrm{tr} \left( O \rho_{\mathrm{avg}}\right) \right| \geq \epsilon \right]
\leq \frac{2^{\mathrm{w}(O)} \mathrm{tr}(O^2)}{N \epsilon^2}. \label{eq:tail-bound}
\end{equation}
There are different ways to capitalize on this tail bound. Firstly, we can fix a desired approximation accuracy $\epsilon$, as well as a desired (maximal) failure probability $\delta$. Eq.~\eqref{eq:tail-bound} then provides us with a lower bound on the number of measurements $N$ required to achieve these values. 

\begin{thm}[general linear error bound]
Fix $\epsilon,\delta \in (0,1)$ and a linear observable $O$. Suppose that we perform
\begin{equation*}
N \geq \frac{2^{\mathrm{w}(O)} \mathrm{tr} (O^2 )}{\epsilon^2 \delta}
\end{equation*}
randomized single-qubit measurements on independent states.
Then, the associated classical shadow 
suffices to $\epsilon$-approximate the expectation value of the average source state:
\begin{equation*}
\left| \hat{o}_{(N)}-\mathrm{tr} \left( O \rho_{\mathrm{avg}}\right) \right| \leq \epsilon
\quad \text{with prob.~(at least) $1-\delta$.}
\end{equation*}
\end{thm}

Alternatively, we can fix a confidence level $\alpha$ and a total measurement budget $N$. In this case, Eq.~\eqref{eq:tail-bound} provides us with a bound on the approximation accuracy. Together with the empirical average $\hat{o}_{(N)}$ itself, this provides a statistically sound confidence interval.

\begin{cor}[general confidence interval]
Fix an observable $O$,  a confidence level $\alpha \in (0,1)$, as well as a measurement budget $N$ (comprised of independent states). Then, the true observable average $\mathrm{tr}(O \rho_{\mathrm{avg}})$ is contained in the interval
\begin{equation*}
\left[\hat{o}_{(N)}-\epsilon,\hat{o}_{(N)}+\epsilon \right] \quad \text{with} \quad \epsilon = \sqrt{\frac{2^{\mathrm{w}(O)} \mathrm{tr}(O^2)}{N (1-\alpha)}}
\end{equation*}
with probability (at least) $\alpha$.
\end{cor}

These statements tell us that the measurement budget $N$ (required number of independent state copies) should scale with $2^{\mathrm{w}(O)}\mathrm{tr}(O^2)$ and the approximation error decays as $1/\sqrt{N}$. 
This is asymptotically optimal because of the central limit theorem,
but the scaling in $1/(1-\alpha)$ is extremely poor. More sophisticated estimation techniques -- like median of means instead of empirical averages \cite{huang2020shadow} -- improve this dependence exponentially from $1/(1-\alpha)$ to $\mathrm{const}\times \log (1/(1-\alpha))$.

\subsection{Predicting quadratic functions with classical shadows}

\subsubsection{U-statistics estimator}

The linear prediction ideas from above do extend to higher order polynomials. 
But in contrast to before, independent, but not identical, state copies (``drifting sources'') do require extra attention.
Here, we restrict our attention to quadratic polynomials \cite{huang2020shadow}. An extension to higher order polynomials is conceptually straightforward, but can become somewhat tedious to analyze \cite{EKH20}. Recall that we can rewrite any quadratic function in $\rho$ as a linear function in the tensor product $\rho \otimes \rho$:
\begin{equation*}
q (\rho) = \mathrm{tr} \left( Q \rho \otimes \rho \right).
\end{equation*}
We can approximate this function by replacing each exact copy of the unknown state with distinct classical snapshots (say $\hat{\rho}_i$ and $\hat{\rho}_j$, with $i \neq j$). Independence of the underlying states ensures stochastic independence of the classical snapshots and we conclude
\begin{equation*}
\mathrm{tr} \left( Q \hat{\rho}_i \otimes \hat{\rho}_j\right) \quad \text{obeys} \quad \mathbb{E} \mathrm{tr} \left( Q \hat{\rho}_i \otimes \hat{\rho}_j \right) = \mathrm{tr} \left( Q \rho_i \otimes \rho_j \right).
\end{equation*}
This is not a good estimator (yet).
We can improve approximation accuracy by empirically averaging over all \emph{distinct} pairs of $N$ classical shadows $\hat{\rho}_1,\ldots,\hat{\rho}_N$:
\begin{equation}
\hat{q}_{(N)} = \tfrac{1}{N(N-1)} \sum_{i \neq j} \mathrm{tr} \left( Q \hat{\rho}_i \otimes \hat{\rho}_j \right). \label{eq:quadratic-prediction}
\end{equation}
This is the simplest example of a \emph{U-statistics estimator} \cite{hoeffding1992class}. 
It is invariant under permuting the individual snapshots $\hat{\rho}_i$ and $\hat{\rho}_j$. This invariance allows us to also symmetrize the quadratic observable. We can without loss assume  $\mathrm{tr}(Q X \otimes Y) = \mathrm{tr} \left( Q Y \otimes X \right)$ for all matrices $X,Y$ with compatible dimension. Such a symmetry will simplify our derivations considerably.

\subsubsection{Deterministic bias}

The estimator average~\eqref{eq:quadratic-prediction} exactly reproduces $q \left(\rho_{\mathrm{avg}}\right)=\mathrm{tr}(Q \rho_{\mathrm{avg}} \otimes \rho_{\mathrm{avg}})$ if and only if the underlying states are identical ($\rho_1 = \cdots = \rho_N = \rho)$. 
If this is not the case, the expectation values of individual classical shadows can be distinct from each other. 
This can introduce a bias when attempting to estimate the average behavior of a quadratic function. Fortunately, any such bias is suppressed by $1/N$ and approaches zero once the number of measurements gets sufficiently large.
This is the content of the following statement. Let $\|\cdot \|_1$ and $\| \cdot \|_\infty$ denote the trace and operator norm, respectively.

\begin{lemma}[quadratic bias for non-identical states] \label{lem:bias}
Let $\left\{\hat{\rho}_1,\ldots,\hat{\rho}_N \right\}$ be a classical shadow that arise from measuring independent states $\rho_1,\ldots,\rho_N$. Set $\rho_{\mathrm{avg}}=\tfrac{1}{N} \sum_{i=1}^N \rho_i$ and consider a quadratic function $q(\sigma) = \mathrm{tr}(Q \sigma \otimes \sigma)$. Then, the associated U-statistics estimator~\eqref{eq:quadratic-prediction} obeys
\begin{equation*}
\mathbb{E} \hat{q}_{(N)}=q(\rho_{\mathrm{avg}}) + \tfrac{\Delta}{N-1}  \text{ with } |\Delta| \leq \max_{1 \leq k \leq N} \|\rho_{\mathrm{avg}}-\rho_k \|_1 \|Q\|_\infty.
\end{equation*}
\end{lemma}

Note that the bias term $\Delta$ vanishes if all states are identically distributed ($\rho_i = \rho_j$ for all $1 \leq i,j \leq N$) and can never be too large either:
\begin{equation*}
|\Delta| \leq 2 \|Q \|_\infty,
\end{equation*}
as the trace norm difference between two quantum states is at most two. Many quadratic functions, which are particularly relevant in entanglement detection, also obey $\|Q \|_\infty \leq 1$.

\begin{proof}[Proof of Lemma~\ref{lem:bias}]

Apply $\mathbb{E} \left[ \hat{\rho}_i \otimes \hat{\rho}_j \right] = \rho_i \otimes \rho_j$ (independence) and elementary reformulations to conclude
\begin{align*}
\Delta=& (N-1)\left( q \left(\rho_{\mathrm{avg}}\right)-\mathbb{E} \hat{q}_N  \right) \\
=& \tfrac{1}{N} \left( \tfrac{N-1}{N} \sum_{i,j=1}^N \mathrm{tr} \left(Q \rho_i \otimes \rho_j \right) - \sum_{i \neq j}\mathrm{tr} \mathbb{E}\left( Q \hat{\rho}_i \otimes \hat{\rho}_j \right)\right) \\
=& \tfrac{1}{N}\left( \sum_{i=1}^N \mathrm{tr} \left( Q \rho_i \otimes \rho_i \right) - \tfrac{1}{N}\sum_{i,j=1}^N \mathrm{tr} \left( Q \rho_i \otimes \rho_j \right) \right) \\
=& \tfrac{1}{N}\sum_{i=1}^N \mathrm{tr} \left(Q \rho_i \otimes \left( \rho_i - \rho_{\mathrm{avg}} \right) \right) .
\end{align*}
Apply matrix Hoelder to this reformulation to obtain a slightly stronger version of the advertised bound:
\begin{align*}
\left|\Delta \right|
\leq & \tfrac{1}{N}\sum_{i=1}^N \|Q \|_\infty \| \rho_i \otimes (\rho_i - \rho_{\mathrm{avg}})\|_1 \\
=& \| Q \|_\infty \tfrac{1}{N}\sum_{i=1}^N \|\rho_i - \rho_{\mathrm{avg}} \|_1,
\end{align*}
because the trace norm is multiplicative under tensor products and quantum states $\rho_i$ satisfy $\|\rho_i \|_1 =1$.
\end{proof}

\subsubsection{Variance bounds}

In analogy to the linear estimator~\eqref{eq:linear-estimator} (empirical average), the U-statistics estimator~\eqref{eq:quadratic-prediction} converges to its expectation value $\mathbb{E} \left[ \hat{q}_{(N)}\right]$.
The variance, which we compute in the following, provides a useful summary parameter for the rate of this convergence.
Use $\mathbb{E} \left[ \hat{\rho}_i \otimes \hat{\rho}_j \right] = \rho_i \otimes \rho_j$ to rewrite the U-statistics variance as
\begin{align*}
\mathrm{Var} \left[ \hat{q}_{(N)}\right]
=& \mathbb{E} \left[ \left( \tfrac{1}{N(N-1)}\sum_{i \neq j}\mathrm{tr} \left( Q \left( \hat{\rho}_i \otimes \hat{\rho}_j - \rho_i \otimes \rho_j \right) \right) \right)^2 \right] \\
=& \tfrac{1}{N^2 (N-1)^2} \sum_{i \neq j} \sum_{k \neq l} \mathbb{E} \left[ \mathrm{tr} \left( Q \left( \hat{\rho}_i \otimes \hat{\rho}_j - \rho_i \otimes \rho_j \right) \right) \right.\\
& \left. \mathrm{tr} \left( Q \left( \hat{\rho}_k \otimes \hat{\rho}_l - \rho_k \otimes \rho_l \right) \right) \right].
\end{align*}
We can now analyze these contributions separately. And, owing to stochastic independence, most of them vanish identically. It is at this point, where the assumption of independent state copies (and access to independent randomness for selecting measurements) matters the most. For instance, if all indices $i,j,k,l$ are distinct, the expectation value factorizes and produces a zero contribution.
As detailed in \cite[Supplemental Material]{EKH20} the only exceptions are contributions where at least two indices coincide. 
Together with symmetry of the observable ($\mathrm{tr} \left(Q X \otimes Y \right) =\mathrm{tr} \left( Q Y \otimes X \right)$) and the AM-GM inequality ($X_j X_k \leq |X_j| |X_k| \leq \tfrac{1}{2} \left( |X_j|^2 + |X_k|^2 \right) = \tfrac{1}{2} \left( X_j^2 + X_k^2 \right)$), we obtain 
\begin{align*}
\mathrm{Var} \left[ \hat{q}_{(N)}\right]\leq & \tfrac{4(N-2)}{N^2(N-1)^2} \sum_{i \neq j} \mathbb{E}  \left[ \mathrm{tr} \left( Q \left( \hat{\rho}_i - \rho_i \right) \otimes \rho_j \right)^2 \right] \\
+& \tfrac{2}{N^2(N-1)^2}\sum_{i \neq j} \mathbb{E} \left[ \mathrm{tr} \left( Q \left( \hat{\rho}_i \otimes \hat{\rho}_j - \rho_i \otimes \rho_j \right) \right)^2 \right] \\
\leq & \tfrac{4(N-2)}{N^2(N-1)^2}\sum_{i \neq j} \mathrm{Var} \left[ \mathrm{tr} \left( \mathrm{tr}_2 \left( Q \mathbb{I} \otimes \rho_j \right) \hat{\rho}_i \right)\right] \\
+& \tfrac{2}{N^2 (N-1)^2}\sum_{i \neq j} \mathrm{Var} \left[ \mathrm{tr} \left( Q \hat{\rho}_i \otimes \hat{\rho}_j \right) \right].
\end{align*}
The final reformulation allows us to re-use the linear variance bound from Lemma~\ref{lem:linear-variance}. For $1 \leq j \leq N$, we define the effective single-copy observable 
\begin{equation}
Q_j = \mathrm{tr}_2 \left(Q \mathbb{I} \otimes \rho_j \right) \label{eq:effective-observable}
\end{equation}
to recognize simple linear variance terms within the first sum.
\begin{align}
\mathrm{Var} \left[ \hat{q}_{(N)}\right]
\leq & \tfrac{4(N-2)}{N^2(N-1)^2}\sum_{i \neq j} 2^{\mathrm{w}(Q_j)} \mathrm{tr} \left(Q_j^2 \right) \nonumber \\
+& \tfrac{2}{N^2 (N-1)^2} \sum_{i \neq j} 2^{\mathrm{w}(Q)} \mathrm{tr}(Q^2) \nonumber \\
\leq & \tfrac{2}{N}\left(2 \max_{1 \leq i \leq N} 2^{\mathrm{w}(Q_i)} \mathrm{tr}(Q_i^2) + \tfrac{1}{N-1}2^{\mathrm{w}(Q)}\mathrm{tr}(Q^2) \right). \label{eq:quadratic-variance}
\end{align}
Clearly, this upper bound becomes smaller as the measurement budget $N$ increases. 

\subsubsection{Error bound and confidence interval}

Having derived bounds on deterministic bias and variance allows us to deduce a general error bound.
Markov's inequality implies
\begin{align*}
& \mathrm{Pr} \left[ \left| \hat{q}_{(N)}-q \left( \rho_{\mathrm{avg}} \right) \right| \geq \epsilon \right]
= \mathrm{Pr} \left[ \left( \hat{q}_{(N)}-q (\rho_{\mathrm{avg}})\right)^2 \geq \epsilon^2 \right] \\
\leq & \tfrac{1}{\epsilon^2}\mathbb{E} \left[ \left( \hat{q}_{(N)}-\mathbb{E} \left[ \hat{q}_{(N)} \right] + \mathbb{E} \left[ \hat{q}_{(N)} \right] - q\left(\rho_{\mathrm{avg}}\right) \right)^2 \right]\\
=& \tfrac{1}{\epsilon^2}\left( \mathbb{E} \left[ \left( \hat{q}_{(N)}-\mathbb{E} \left[ \hat{q}_{(N)}\right] \right)^2 \right] + \left( \mathbb{E} \left[ \hat{q}_{(N)} \right] - q \left( \rho_{\mathrm{avg}} \right) \right)^2\right) \\
=& \tfrac{1}{\epsilon^2}\left( \mathrm{Var} \left[ \hat{q}_{(N)}\right] + \frac{\Delta^2}{(N-1)^2}\right),
\end{align*}
where we have isolated statistical fluctuations from the underlying deterministic bias. Inserting the bounds from Eq.~\eqref{eq:quadratic-variance} and Lemma~\ref{lem:bias} renders this bound more explicit:
\begin{align}
& \mathrm{Pr} \left[ \left| \hat{q}_{(N)}-q \left( \rho_{\mathrm{avg}} \right) \right| \geq \epsilon \right] \label{eq:main-bound} \\
\leq & 4 \max_{1 \leq i \leq N}\frac{2^{\mathrm{w}(Q_i)} \mathrm{tr}(Q_i^2)}{\epsilon^2 N} + 2\frac{2^{\mathrm{w}(Q)}\mathrm{tr}(Q^2)}{N (N-1)\epsilon^2}+ 4 \frac{ \|Q \|_\infty^2}{(N-1)^2 \epsilon^2}.
\nonumber
\end{align}
Term two and three have a comparable scaling in $N$ and $\epsilon$. The first term is different and starts to dominate as $N$ increases.
Recall that $Q (\rho_i) = \mathrm{tr} \left( Q \mathbb{I} \otimes \rho_i \right)$ denotes an effectively linear function on a single copy of $AB$. This is the power of U-statistics. Asymptotically, it is an effectively linear scaling term that dominates the statistical convergence rate of a quadratic estimator. 
The other terms, however, can dominate in the small-$N$ regime. 

\section{Concrete guarantees for estimating $D_2$}

A direct conversion into error bound and confidence interval is conceptually straightforward, but somewhat cumbersome. Different contributions with distinct scaling behavior must be balanced against each other. This renders fully general statements somewhat difficult to parse. Instead, we derive concrete error bounds and confidence intervals for a concrete, and interesting, quadratic polynomial, namely $D_2$.

\subsection{Rewriting $D_2$ as a quadratic function}

As mentioned in the main text, the $D_2$ PPT condition is trivially satisfied by the partial transpose of any density matrix. Nevertheless, its symmetry-resolved counterpart provides a non-trivial entanglement condition, and is, from an experimental point of view, the most affordable entanglement condition we proposed. Indeed, it requires to estimate only the first two moments of the partial transpose within a sector. For this reason, we consider here states possessing the same symmetry as in Sec.~\ref{sec.sr} of the main text. That is, we focus here on a $(n+m)$-qubit mixed state $\rho$ that we view as bipartite states, with a subsystem $A$ containing $n$ qubits and subsystem $B$ containing $m$ qubits. Moreover, we assume that $\rho$ commutes with the total
number operator $\mathcal{N}=\mathcal{N}_A+\mathcal{N}_B$.  

Let us recall here that such a state $\rho$ has the block diagonal structure $\rho = \sum_i Q_i \rho Q_i$ and that its partial transpose is also block diagonal, albeit in a different basis: $\rho^\Gamma = \sum_i P_i \rho^\Gamma P_i$ (see Eqs.~\eqref{def:Q} and~\eqref{def:P} for the definitions of the projectors $Q_i$ and $P_i$). We fix a block label $i$
and consider the second moment inequality restricted to this block: 
\begin{equation*}
p_2 \left( P_i \rho^\Gamma P_i \right) \leq \left( p_1 \left( P_i \rho^\Gamma P_i \right) \right)^2.
\end{equation*}
This inequality is true if and only if the following homogeneous polynomial is nonnegative:
\begin{align*}
D_2^{(i)}(\rho) = \left( \mathrm{tr} \left( P_i \rho^\Gamma P_i \right) \right)^2 - \mathrm{tr} \left( P_i \rho^\Gamma P_i \right)^2. 
\end{align*}
Lemma~\ref{lem.Li} in the main text allows us to rewrite this expression as a linear and symmetric function in $\rho \otimes \rho$:
\begin{align}
D_2^{(i)}(\rho) =& \mathrm{tr} \left( Q \rho \otimes \rho \right) \label{eq:D2-definition}
\end{align}
where $Q$ is a linear operator that acts on two copies of the bipartite operator space $\mathcal{B}(\mathcal{H}_{AB}) \otimes \mathcal{B}(\mathcal{H}_{AB}) \simeq \mathcal{B}(\mathcal{H}_A) \otimes \mathcal{B}(\mathcal{H}_B) \otimes \mathcal{B} (\mathcal{H}_A) \otimes \mathcal{B}(\mathcal{H}_B)$.
It is defined in terms of the orthogonal projectors 
\begin{align*}
\Pi_a =& 
\sum_{i_1+\cdots+i_n=a} |i_1,\ldots,i_n \rangle \! \langle i_1,\ldots,i_n| \in \mathcal{B}(\mathcal{H}_A), \\
\Pi_b =& 
\sum_{i_1 + \cdots + i_m = b} |i_1,\ldots,i_m \rangle \! \langle i_1,\ldots,i_m| \in \mathcal{B}(\mathcal{H}_B),
\end{align*}
as well as different swap operations.
Let $W_A$ ($W_B$) and $W_{AB}$ denote the operator that swaps both copies of $A$ ($B$) and $AB$, respectively.
Then,
\begin{align}
Q 
=& P_i \otimes P_i - \tfrac{1}{2} W_A (P_i \otimes \mathbb{I}_{AB}) W_A - \tfrac{1}{2} W_B (P_i \otimes \mathbb{I}_{AB}) W_B \label{eq:Q}
\end{align}
obeys Eq.~\eqref{eq:D2-definition} in a symmetric fashion, i.e.\ 
$
\mathrm{tr} \left( Q \rho \otimes \sigma \right) = \mathrm{tr} \left( Q \sigma \otimes \rho \right)$ for $\rho,\sigma \in \mathcal{B}(\mathcal{H}_{AB})$.
We can use this two-copy operator to construct a U-statistics estimator based on (independent) classical shadows:
\begin{equation}
\widehat{D}_{2,(N)}^{(i)}= \tfrac{1}{N(N-1)}\sum_{i \neq j} \mathrm{tr} \left( Q \hat{\rho}_i \otimes \hat{\rho}_j \right).
\label{eq:D2-estimator}
\end{equation}
Due to Lemma~\ref{lem:bias},
\begin{align*}
\mathbb{E} \widehat{D}_{2,(N)}^{(i)} \overset{N \to \infty}{\longrightarrow} \mathrm{tr} \left( Q \rho_{\mathrm{avg}}^{\otimes 2}\right) = D_2^{(i)}\left( \rho_{\mathrm{avg}}\right)
\end{align*}
and Eq.~\eqref{eq:main-bound} allows us to rigorously control the speed of convergence.

\subsection{Bounding the relevant figures of merit}

The bound provided by Eq.~\eqref{eq:main-bound} depends on the squared Hilbert-Schmidt norms of $Q$, as well as the squared Hilbert Schmidt norm of
\begin{align}
Q (\sigma) =& \mathrm{tr}_2 \left( Q \mathbb{I}\otimes \sigma \right) \label{eq:Q-linear}\\
=& \mathrm{tr} \left( P_i \sigma \right) P_i - \tfrac{1}{2} \sum_{a-b=i} \Pi_a \otimes \mathbb{I}_B \sigma \mathbb{I}_A \otimes \Pi_b \nonumber \\
-& \tfrac{1}{2}\sum_{a-b=i} \mathbb{I}_A \otimes \Pi_b \sigma \Pi_a \otimes \mathbb{I}_B \in \mathcal{B}(\mathcal{H}_{AB}), \nonumber 
\end{align}
where $\sigma$ may range over all possible source states $\rho_1,\ldots,\rho_N$.
The weights (localities) of $Q$ and $Q (\sigma)$ also play an important role. 
Alas, these are essentially as large as they can be. Viewed as observables on one and two copies of $AB$, respectively, the operators $Q(\sigma) \in \mathcal{B}(\mathcal{H}_{AB})$ and $Q \in \mathcal{B}(\mathcal{H}_{AB})^{\otimes 2}$ act nontrivially on all qubits involved. Hence, we must conclude
\begin{align}
\mathrm{w}\left( Q (\sigma) \right) = |AB|=(n+m) \quad \text{and} \label{eq:D2-reduced-weight} \\
\mathrm{w}(Q) = |ABAB| =2 (n+m). \label{eq:D2-full-weight}.
\end{align}
Next, note that 
the sum $\sum_{a-b=i} \Pi_a \otimes \Pi_b$ describes an orthogonal projection $P_i$ on the $AB$-system. (Tensor products of) orthogonal projectors have operator norm one ($\|P_i \|_\infty =1$). Operator norms are also invariant under permuting tensor factors. Therefore,
\begin{align}
\| Q \|_\infty  \leq & \| P_i \otimes P_i \|_\infty 
+ \tfrac{1}{2}\| W_B P_i \otimes \mathbb{I}_{AB} W_B\|_\infty \nonumber \\
+& \tfrac{1}{2}\| W_A P_i \otimes \mathbb{I}_{AB} W_A \|_\infty \nonumber \\
=& \|P_i \otimes P_i \|_\infty + \| P_i \otimes \mathbb{I}_{AB}\|_\infty \label{eq:D2-operator}
=2.
\end{align}

Let us now move on to computing the Hilbert-Schmidt norms (squared) of $Q$ and $Q(\sigma)$.

\begin{lemma} \label{lem:D2-quadratic}
The two-copy observable defined in Eq.~\eqref{eq:Q} obeys
\begin{equation*}
\mathrm{tr}(Q^2) \leq \tfrac{1}{2} \mathrm{tr}(P_i) \left( 3 \mathrm{tr}(P_i) + 2^{|AB|} - 4 \right).
\end{equation*}
\end{lemma}

\begin{proof}
To simplify notation somewhat, we introduce the following short-hand notation conventions: 
$\tilde{\Pi}_a = \Pi_a \otimes \mathbb{I}_B \in \mathcal{B} \left( \mathcal{H}_{AB}\right)$
and
$\tilde{\Pi}_b = \mathbb{I}_A \otimes \Pi_b \in \mathcal{B} \left( \mathcal{H}_{AB}\right)$.
This notation allows us to split up the Hilbert-Schmidt norm squared into four contributions:
\begin{align*}
\mathrm{tr}(Q^2)
=& \mathrm{tr} \left(( P_i \otimes P_i - \tfrac{1}{2}\sum_{a-b=i} \left(\tilde{\Pi}_a \otimes \tilde{\Pi}_b + \tilde{\Pi}_b\otimes \tilde{\Pi}_a \right)\right)^2 \\
=& \mathrm{tr} \left( P_i^2 \right)^2 
-2  \sum_{a-b=i} \mathrm{tr} \left( P_i \tilde{\Pi}_a\right) \mathrm{tr} \left(P_i\tilde{\Pi}_b \right) 
\\
+& \tfrac{1}{2} \sum_{a-b=i} \sum_{a'-b'=i} \mathrm{tr} \left( \tilde{\Pi}_a \tilde{\Pi}_{a'} \right) \mathrm{tr} \left(\tilde{\Pi}_b \tilde{\Pi}_{b'} \right) \\
+ & \tfrac{1}{2} \sum_{a-b=i} \sum_{a'-b'=i} \mathrm{tr} \left( \tilde{\Pi}_a \tilde{\Pi}_{b'} \right) \mathrm{tr} \left( \tilde{\Pi}_b \tilde{\Pi}_{a'} \right).
\end{align*}
We can use $P_i^2 =P_i$, $P_i = \sum_{a'-b'=i} \Pi_{a'} \otimes \Pi_{b'}$ and orthogonality relations to simplify terms individually. In particular,
\begin{align*}
&  \sum_{a-b=i} \mathrm{tr}(P_i \tilde{\Pi}_a ) \mathrm{tr} (P_i\tilde{\Pi}_b) \\
=& \sum_{a-b=i} \sum_{a'-b'=i} \sum_{a''-b''=i} 
\mathrm{tr} \left(\Pi_{a'} \Pi_a \otimes \Pi_{b'} \right) \mathrm{tr} \left( \Pi_{a''} \otimes \Pi_{b''} \Pi_b \right) \\
=&  \sum_{a-b=i} \sum_{a'-b'=i} \sum_{a''-b''=i} \delta_{a,a'} \delta_{b,b''} \mathrm{tr}(\Pi_a) \mathrm{tr}(\Pi_{b'}) \mathrm{tr}(\Pi_{a''}) \mathrm{tr}(\Pi_b) \\
=&  \sum_{a-b=i} \mathrm{tr}(\Pi_a)^2 \mathrm{tr}(\Pi_b)^2,
\end{align*}
and \begin{align*}
& \sum_{a-b=i} \sum_{a'-b'=i}  \mathrm{tr} \left( \tilde{\Pi}_a \tilde{\Pi}_{a'} \right) \mathrm{tr} \left( \tilde{\Pi}_b \tilde{\Pi}_{b'} \right) \\
=& \sum_{a-b=i} \sum_{a'-b'=i}  \delta_{a,a'} \delta_{b,b'} \mathrm{tr}(\Pi_a \otimes \mathbb{I}_B ) \mathrm{tr}(\mathbb{I}_A \otimes \Pi_b) \\
=& \sum_{a-b=i} 2^{|A|} 2^{|B|}\mathrm{tr} \left( \Pi_a \otimes \Pi_b \right) = 2^{|AB|}\mathrm{tr}(P_i),
\end{align*}
as well as
\begin{align*}
& \sum_{a-b=i} \sum_{a'-b'=i} \mathrm{tr} \left( \tilde{\Pi}_a \tilde{\Pi}_{b'} \right) \mathrm{tr} \left( \tilde{\Pi}_b \tilde{\Pi}_{a'} \right) \\
=&  \sum_{a-b=i} \sum_{a'-b'=i} \mathrm{tr} \left( \tilde{\Pi}_a \right) \mathrm{tr} \left( \Pi_{b'} \right) \mathrm{tr} \left( \Pi_{a'} \right) \mathrm{tr} \left( \Pi_b \right) \\
=& \left( \sum_{a-b=i} \mathrm{tr}(\Pi_a) \mathrm{tr}(\Pi_b) \right)
\left( \sum_{a'-b'=i} \mathrm{tr} \left( \Pi_{a'} \right) \mathrm{tr} \left( \Pi_{b'} \right)\right) \\
=& \left( \sum_{a-b=i} \mathrm{tr} \left( \Pi_a \otimes \Pi_b \right) \right)^2 = \mathrm{tr} \left( P_i \right)^2
\end{align*}
\begin{align*}
\mathrm{tr}(Q^2)
=& \tfrac{3}{2}\mathrm{tr} \left( P_i \right)^2 + \tfrac{1}{2} 2^{|AB|} \mathrm{tr} \left( P_i \right)  \\
-& 2  \sum_{a-b=i}  \mathrm{tr}(\Pi_a)^2 \mathrm{tr}(\Pi_b)^2.
\end{align*}
The remaining sum can simply be bounded by observing that $\mathrm{tr}(\Pi_a)^2 \geq \mathrm{tr}(\Pi_a) \geq 0$ and similarly for $\Pi_b$:
\begin{equation*}
 \sum_{a-b=i} \mathrm{tr}(\Pi_a)^2 \mathrm{tr}(\Pi_b)^2
 \geq \sum_{a-b=i} \mathrm{tr}(\Pi_a) \mathrm{tr}(\Pi_b) = \mathrm{tr}(P_i).
 \end{equation*}
 This implies the following upper bound on the quadratic variance:
 \begin{align*}
\mathrm{tr}(Q^2)
 \leq &\tfrac{3}{2} \mathrm{tr}(P_i)^2 + \tfrac{1}{2} 2^{|AB|} \mathrm{tr}(P_i) - 2 \mathrm{tr}(P_i) \\
 =& \tfrac{1}{2} \mathrm{tr}(P_i) \left( 3 \mathrm{tr}(P_i) + 2^{|AB|} - 4 \right).
 \end{align*}
\end{proof}

\begin{lemma} \label{lem:D2-linear}
The single-copy observable defined in Eq.~\eqref{eq:Q-linear} obeys
\begin{align*}
\mathrm{tr}\left(Q(\sigma)^2\right) = & \mathrm{tr}(P_i \sigma)^2 \mathrm{tr}(P_i) - 2 \mathrm{tr}(P_i \sigma)^2   \\
+& \tfrac{1}{2} \mathrm{tr} \left(\sigma^\Gamma P_i \sigma^\Gamma \right)+ \tfrac{1}{2}\mathrm{tr} \left( P_i \sigma^\Gamma P_i \sigma^\Gamma \right). 
\end{align*}
\end{lemma}

\begin{proof}
Use the notation introduced in the previous proof to rewrite
\begin{equation*}
Q(\sigma) = \mathrm{tr}(P_i \sigma) P_i - \tfrac{1}{2} \sum_{a-b=i}\left( \tilde{\Pi}_a \sigma \tilde{\Pi}_b + \tilde{\Pi}_b \sigma \tilde{\Pi}_a \right)
\end{equation*}
and, in turn,
\begin{align*}
\mathrm{tr} \left( Q(\sigma)^2 \right)
=&  \mathrm{tr}(P_i \sigma)^2 \mathrm{tr}(P_i^2) - \mathrm{tr}(P_i \sigma) \sum_{a-b=i} \mathrm{tr} \left( P_i \tilde{\Pi}_a \sigma \tilde{\Pi}_b \right) \\
-& \mathrm{tr}(P_i \sigma) \sum_{a-b=i} \mathrm{tr} \left( P_i \tilde{\Pi}_b \sigma \tilde{\Pi}_a \right) \\
+& \tfrac{1}{4} \sum_{a-b=i} \sum_{a'-b'=i} \mathrm{tr} \left( \tilde{\Pi}_a \sigma \tilde{\Pi}_b \tilde{\Pi}_{a'} \sigma \tilde{\Pi}_{b'} \right) \\
+& \tfrac{1}{2} \sum_{a-b=i} \sum_{a'-b'=i} \mathrm{tr} \left( \tilde{\Pi}_a \sigma \tilde{\Pi}_b \tilde{\Pi}_{b'} \sigma \tilde{\Pi}_{a'} \right) \\
+& \tfrac{1}{4} \sum_{a-b=i} \sum_{a'-b'=i} \mathrm{tr} \left( \tilde{\Pi}_b \sigma \tilde{\Pi}_a \tilde{\Pi}_{b'} \sigma \tilde{\Pi}_{a'} \right).
\end{align*}
These terms can be simplified individually. The first term is already almost as simple as possible. For the second term, use $P_i = \sum_{a-b=i}\Pi_a \otimes \Pi_b$, as well as orthogonality relations among the projectors ($\Pi_a \Pi_{a'} = \delta_{a,a'} \Pi_a$ and $\Pi_b \Pi_{b'}=\delta_{b,b'}\Pi_b$), to conclude
\begin{align*}
\sum_{a-b=i} \mathrm{tr} \left( P_i \tilde{\Pi}_a \sigma \tilde{\Pi}_b \right)
=& \sum_{a-b=i} \mathrm{tr} \left( \tilde{\Pi}_b P_i \tilde{\Pi}_a \sigma \right) 
= \mathrm{tr}(P_i \sigma ) .
\end{align*}
The third term is exactly the same.
Terms four and six are also equivalent, because of commutation relations ($\tilde{\Pi}_a \tilde{\Pi}_b = \tilde{\Pi}_b \tilde{\Pi}_a$) and cyclicity of the trace. We can use some partial transpose tricks to obtain:
\begin{align*}
& \tfrac{1}{4}\sum_{a-b=i} \sum_{a'-b'=i} \left(\mathrm{tr} \left( \tilde{\Pi}_a \sigma \tilde{\Pi}_b \tilde{\Pi}_{a'} \sigma \tilde{\Pi}_{b'} \right)
+ \mathrm{tr} \left( \tilde{\Pi}_b \sigma \tilde{\Pi}_a \tilde{\Pi}_{b'} \sigma \tilde{\Pi}_{a'} \right) \right)\\
 =&\tfrac{1}{2} \sum_{a-b=i} \sum_{a'-b'=i} \mathrm{tr} \left(\Pi_a \otimes \Pi_{b'} \sigma \Pi_{a'} \otimes \Pi_b \sigma \right) \\
 =& \tfrac{1}{2}\sum_{a-b=i} \sum_{a'-b'=i} \mathrm{tr} \left(
\left( \Pi_{a'} \otimes \Pi_{b'} \sigma^\Gamma \Pi_a \otimes \Pi_b \right)^\Gamma \sigma \right) \\
=& \tfrac{1}{2}\mathrm{tr} \left( P_i \sigma^\Gamma P_i \sigma^\Gamma \right).
\end{align*}
Similar tricks apply to the fifth term:
\begin{align*}
& \tfrac{1}{2} \sum_{a-b=i} \sum_{a'-b'=i} \mathrm{tr} \left( \tilde{\Pi}_a \sigma \tilde{\Pi}_b \tilde{\Pi}_{b'} \sigma \tilde{\Pi}_{a'} \right) \\
=& \tfrac{1}{2} \sum_{a-b=i} \sum_{a'-b'=i} \delta_{a,a'} \delta_{b,b'} \mathrm{tr} \left( \Pi_a \otimes \mathbb{I}_B \sigma \mathbb{I}_A \otimes \Pi_b \sigma \right) \\
=& \tfrac{1}{2} \sum_{a-b=i} \mathrm{tr} \left( \left( \mathbb{I}_A \otimes \mathbb{I}_B \sigma^\Gamma \Pi_a \otimes \Pi_b \right)^\Gamma \sigma \right) \\
=& \tfrac{1}{2} \mathrm{tr} \left( \sigma^\Gamma P_i \sigma^\Gamma \right).
\end{align*}
The advertised bound now follows from putting everything together.
\end{proof}

\subsection{$D_2$ error bounds and confidence intervals}

We can combine the weight estimates \eqref{eq:D2-full-weight}, \eqref{eq:D2-reduced-weight} with the operator norm bound from Eq.~\eqref{eq:D2-operator} and the squared Hilbert-Schmidt norm bounds from Lemma~\ref{lem:D2-quadratic}, as well as Lemma~\ref{lem:D2-linear}, and insert them into the general error bound from Eq.~\eqref{eq:main-bound}. To simplify later computations, we also replace $1/N$-factors by $1/(N-1)$. For accuracy $\epsilon \in (0,1)$, we obtain
\begin{align*}
& \mathrm{Pr} \left[ \left| \widehat{D}^{(i)}_{2,(N)} - D_2^{(i)}\left( \rho_{\mathrm{avg}}\right) \right| \geq \epsilon \right] \\
\leq & 4 \max_{1 \leq k \leq N}\tfrac{2^{\mathrm{w}(Q(\rho_k))} \mathrm{tr}(Q(\rho_k)^2)}{\epsilon^2 (N-1)} + 2\tfrac{2^{\mathrm{w}(Q)}\mathrm{tr}(Q^2)}{ (N-1)^2\epsilon^2}+ 4 \tfrac{  \|Q \|_\infty^2}{(N-1)^2 \epsilon^2}
\nonumber \\
\leq & \tfrac{2^{n+m} \mathrm{tr}(P_i)}{\epsilon^2 (N-1)}\max_{1 \leq k \leq N} 4 \left( \mathrm{tr}(P_i \rho_k)^2  + \tfrac{ \mathrm{tr} \left( \rho^\Gamma_k P_i \rho^\Gamma_k \right) - 2 \mathrm{tr}(P_i \rho_k)^2 }{ \mathrm{tr}(P_i)} \right) \\
+& \tfrac{2^{3(n+m)}\mathrm{tr}(P_i)}{(N-1)^2 \epsilon^2}
\left(1 + \tfrac{3 \mathrm{tr}(P_i)  -4 + 8/(\mathrm{tr}(P_i) 2^{2(n+m)})  }{2^{n+m}}\right).
\end{align*}
Here, the first terms in each parenthesis are the leading contributions.
The two remaining factors are of order one and do depend on the particular problem in question. Assuming $ \mathrm{tr}(P_i) \geq 2$ (non-trivial symmetry sector) and $n+m \geq 4$ (at least 4 qubits) ensures
\begin{align}
C_1 :=& \max_{1 \leq k \leq N} 4 \left( \mathrm{tr}(P_i \rho_k)^2  + \tfrac{ \mathrm{tr} \left( \rho^\Gamma_k P_i \rho^\Gamma_k \right) - 2 \mathrm{tr}(P_i \rho_k)^2 }{ \mathrm{tr}(P_i)} \right)  \notag \\
& \leq 4 \left(1- \tfrac{1}{\mathrm{tr}(P_i)} \right),  \label{eq:D2-C1} \\
C_2 :=& \left(1 + \tfrac{3 \mathrm{tr}(P_i)  -4 + 8/(\mathrm{tr}(P_i) 2^{2(n+m)})  }{2^{n+m}}\right)
\leq 2 \label{eq:D2-C2}.
\end{align}
In order to get an error bound, we need to answer the following question: How large does $N$ have to be in order to ensure that this upper bound does not exceed $\delta$?
This is equivalent to demanding
\begin{equation}
(N-1)^2  \geq  (N-1)  C_1 \frac{2^{n+m}\mathrm{tr}(P_i)}{\epsilon^2 \delta} 
+ C_2 \frac{2^{3(n+m)}\mathrm{tr}(P_i)}{\epsilon^2 \delta} \label{eq:final-relation}
\end{equation}
and can be answered by solving a quadratic equation in $(N-1)$.
Doing so implies the main result of this appendix section.

\begin{thm}[Error bound for $D_2$] \label{thm:D2}
Fix $\epsilon,\delta \in (0,1)$, a bipartition $AB$, as well as a symmetry sector $i$.
Let $C_1,C_2$ be the problem-dependent constants introduced in Eqs.~\eqref{eq:D2-C1},\eqref{eq:D2-C2} and suppose that we perform
\begin{align*}
N \geq \frac{2^{n+m}\mathrm{tr}(P_i)}{\epsilon^2 \delta} \frac{1}{2}
\left(
C_1 + \sqrt{C_1^2 + C_2 \frac{\epsilon^2 \delta 2^{n+m}}{\mathrm{tr}(P_i)}}
\right)+1
\end{align*}
randomized, single qubit measurements on independent states $\rho_1,\ldots,\rho_N$.
Then, the $D_2$-estimator~\eqref{eq:D2-estimator} obeys
\begin{equation*}
\left| \widehat{D}^{(i)}_{2,(N)}-D_2^{(i)}(\rho_{\mathrm{avg}}) \right| \leq \epsilon \quad \text{with prob.\ (at least) $1-\delta$.}
\end{equation*}
\end{thm}

This error bound addresses the estimation of $D_2^{(i) }(\rho_{\mathrm{avg}})$ in terms of a single U-statistics estimator. 
The poor scaling in $1/\delta$ can be exponentially improved by dividing the classical shadow into equally-sized batches and performing a median-of-U-statistics estimation instead \cite{huang2020shadow}, which reduced the scaling with $1/\delta$ to a scaling with $\text{const} \times \log (1/\delta)$. However, numerical experiments conducted in Ref.~\cite{EKH20} suggest that this trade-off is only worthwhile if one attempts to predict many properties with the same data set.

Alternatively, we can also fix a 
maximum failure probability $\delta$ and a total measurement budget $N$. Reformulating Rel.~\eqref{eq:final-relation} then provides us with an upper bound on the (squared) approximation accuracy.
This provides us with a statistically sound confidence interval around the estimated polynomial. 

\begin{cor}[confidence interval for $D_2$] \label{cor:D2}
Fix a bipartition $AB$, a symmetry sector $i$, a confidence level $\delta \in (0,1)$ and a measurement budget $N$ (comprised of independent states). Then, 
with probability (at least) $1-\delta$,
\begin{align*}
D_2^{(i)} \left( \rho_{\mathrm{avg}}\right) \in &
\left[ \widehat{D}^{(i)}_{2,(N)}-\epsilon, \widehat{D}^{(i)}_{2,(N)}+\epsilon \right], \quad\text{where} \\
\epsilon =&\sqrt{\frac{2^{n+m}\mathrm{tr}(P_i)}{\delta (N-1)} \left( C_1 + C_2 \frac{2^{2(n+m)}}{N-1}\right)}.
\end{align*}
\end{cor}

Again, median-of-U-statistics estimation allows for improving the dependence on $1/\delta$ exponentially at the cost of a extra constant (Ref.~\cite{huang2020shadow}, for instance, achieves $\mathrm{const} \approx 68$).

\bibliography{Bibliography_EntSymm,Biblio_BV}

\end{document}